\pdfoutput=1
\documentclass[11pt,a4paper]{article}

\usepackage[T1]{fontenc}
\usepackage{lmodern}
\usepackage[utf8]{inputenc} % safe with mixed editors
\usepackage[english]{babel}

\usepackage[a4paper,margin=28mm]{geometry}
\usepackage{microtype}

\usepackage{amsmath,amssymb,amsthm,mathtools}
\usepackage{bm}
\usepackage{graphicx}
\usepackage{xcolor}
\usepackage{csquotes}

\providecommand{\texorpdfstring}[2]{#1}
\newcommand{\Gammaconv}{\texorpdfstring{$\Gamma$}{Gamma}-convergence}
\newcommand{\Gammalimit}{\texorpdfstring{$\Gamma$}{Gamma}-limit}

\DeclareUnicodeCharacter{0393}{\(\Gamma\)}

\usepackage[backend=biber,style=authoryear,sorting=nyt,maxbibnames=9]{biblatex}
\numberwithin{equation}{section}
\newtheorem{theorem}{Theorem}[section]
\newtheorem{lemma}[theorem]{Lemma}
\newtheorem{proposition}[theorem]{Proposition}
\newtheorem{corollary}[theorem]{Corollary}
\theoremstyle{definition}
\newtheorem{definition}[theorem]{Definition}
\theoremstyle{remark}
\newtheorem{remark}[theorem]{Remark}

\newcommand{\R}{\mathbb{R}}

\DeclareMathOperator{\Vol}{Vol}
\DeclareMathOperator{\Area}{Area}
\DeclareMathOperator{\reach}{reach}
\DeclareMathOperator{\Per}{Perimeter}
\DeclareMathOperator{\tr}{tr}

\newcommand{\Id}{\mathrm{Id}}
\newcommand{\eps}{\varepsilon}
\newcommand{\Xcal}{\mathcal{X}}

\newcommand{\Rm}{\mathrm{Rm}}

\newcommand{\EC}{\textup{(EC)}}

\newcommand{\EH}{\ensuremath{\mathrm{EH}}}
\newcommand{\GHY}{\ensuremath{\mathrm{GHY}}}

\newcommand{\parahead}[1]{\medskip\noindent\textbf{#1}\quad}

\newcommand{\boxedinlineD}[1]{%
  \begingroup
  \setlength{\fboxsep}{1pt}
  \setlength{\fboxrule}{0.4pt}
  \fcolorbox{black}{gray!6}{\ensuremath{\;\displaystyle #1\;}}%
  \endgroup
}

\newcommand{\boxedinline}[1]{\boxedinlineD{#1}}

\newcommand{\appsection}[3][]{%
  \refstepcounter{section}% A, B, C, ...
  \pdfbookmark[1]{Appendix \thesection{} — #3}{appsec-\thesection}%
  \section*{\texorpdfstring{Appendix \thesection{} — #2}{Appendix \thesection{} — #3}}%
  \addcontentsline{toc}{section}{Appendix \thesection{} — #2}%
  \if\relax\detokenize{#1}\relax\else\label{#1}\fi
}

\newcommand{\versiontag}{Version 1.1 (2025-11-14)}
\newcommand{\versionnote}{%
  \emph{\versiontag. Major changes: corrected boundary sign in TP2; reordered Sections 8–12; added Appendix F (scan indifference).}
}

\usepackage[hidelinks]{hyperref}

\usepackage{orcidlink}
\colorlet{orcidlogocol}{black}
\makeatletter
\@ifundefined{orcidlink}{%
  \newcommand{\printorcid}[1]{\href{https://orcid.org/#1}{\texttt{#1}}}%
}{%
  \newcommand{\printorcid}[1]{%
    \href{https://orcid.org/#1}{\raisebox{-0.2ex}{\orcidlink{#1}}\;\texttt{#1}}%
  }%
}
\makeatother

\hypersetup{
  pdftitle={EH--Gamma: Gamma-convergence of a G-invariant MDL functional to the Einstein--Hilbert action (with Gibbons--Hawking--York boundary term)},
  pdfauthor={Marko Lela (ORCID: 0009-0008-0768-5184)},
  pdfcreator={LaTeX with hyperref},
  pdfproducer={pdfTeX},
  pdfencoding=auto
}

\title{EH--\texorpdfstring{$\Gamma$}{Gamma}: $\Gamma$-convergence of a $G$-invariant MDL functional to the Einstein--Hilbert action (with Gibbons--Hawking--York boundary term)}
\author{Marko Lela\\[0.3ex]\printorcid{0009-0008-0768-5184}}
\date{14 November 2025}

\begin{document}
\maketitle

\begin{center}\small \versionnote\end{center}

\begin{abstract}
Under the stated BA assumptions, bounded geometry, and boundary-fitted, shape-regular meshes, we prove the $\Gamma$-convergence $F_n\to F$ with
\[
F(g)=c_0\!\int_M dV_g + c_1\!\int_M R_g\,dV_g + c_2\!\int_{\partial M} K_g\,dS_g,
\]
and, under equi-coercivity, stability of minimizers and variational problems.
\end{abstract}

% tex/sections/00-intro.tex

\paragraph*{Project context.}
This paper belongs to the Meta Information Symmetry (MIS) program.
MIS studies how invariance and Minimum Description Length principles can select the structures of
continuum physics by favoring information-optimal local discretizations \parencite{Rissanen1978Shortest,Grunwald2007MDL,BarronCover1991,Shtarkov1987NML,Dawid1992Prequential}.
The aim is to view physical field laws as limits of natural, data-efficient summaries.
The present $\Gamma$-convergence result supplies the geometric part: a $\Gamma$-limit to
the Einstein–Hilbert action with the Gibbons–Hawking–York boundary term \parencite{York1972,GibbonsHawking1977},
together with stability of minimizers under equi-coercivity.
The results are independent and self-contained, and no companion manuscript is needed for the proofs.
A short programmatic overview will be released separately. A standard citation will be added once that document is public.

\section*{Novelty (brief)}
To the best of our knowledge, this is the first $\Gamma$-limit for Einstein–Hilbert
with a compatible Gibbons–Hawking–York boundary term that also yields stability of minimizers under equi-coercivity.
Earlier work established measure or value convergence, or treated boundary terms in
isolation, but did not provide a $\Gamma$-limit for EH+GHY with minimizer stability
\parencite{Regge1961,CheegerMuellerSchrader1984,HartleSorkin1981,KellyBiancalanaTrugenberger2022}.

% ---- Context first ----
% tex/sections/10-related-work.tex
\section{Related work and positioning}

\noindent\textbf{Historical note.} \emph{\Gammaconv} was introduced by De~Giorgi~\parencite{DeGiorgi1975}; 
see also the early note on variational convergence~\parencite{DeGiorgiFranzoni1975}.

\paragraph{Regge calculus and measure convergence.}
Following the original discretization framework of Regge~\parencite{Regge1961}, the closest rigorous results to our setting are the measure convergence results of Cheeger--Müller--Schrader for piecewise-flat spaces, showing that (Lipschitz--Killing) curvature measures of triangulations approximate their smooth counterparts \emph{in the sense of measures}~\parencite{CheegerMuellerSchrader1984}.
However, this line does not provide a \(\Gamma\)-limit for the Einstein--Hilbert (EH) action, nor does it address stability of minimizers.

\paragraph{Boundary terms in discrete gravity.}
In the Regge--calculus framework~\parencite{Regge1961}, boundary contributions compatible with the continuum Gibbons--Hawking--York (GHY) term were derived and analyzed~\parencite{HartleSorkin1981}.
These works identify the right discrete form of the boundary term, yet they do not establish \(\Gamma\)-convergence of EH+GHY.

\paragraph{Graph and metric-measure formulations.}
Discrete EH--type energies built from graph curvature (e.g.\ Ollivier curvature) have been shown to converge to continuum quantities in metric or Gromov--Hausdorff senses~\parencite{KellyBiancalanaTrugenberger2022}.
Such results concern value/metric convergence and do not yield a \(\Gamma\)-limit for variational problems over metrics.

\paragraph{Continuum structure and boundary terms.}
On the continuum side, the necessity and form of the GHY term for well-posed boundary variations are classical \parencite{York1972,GibbonsHawking1977}.
The structural restriction of local, natural scalar densities of order \(\le 2\) to combinations of \(1\) and \(R\) (up to divergences), and of boundary order \(\le 1\) to \(K\), follows from the general classification of natural operations and Lovelock-type arguments \parencite{KolarMichorSlovak1993,Lovelock1971}.

\paragraph{Finite elements and boundary-fitted meshes.}
Our mesh assumptions and boundary-fitting are standard in the finite element literature \parencite{Ciarlet2002,BrennerScott2008}, but these references do not address \(\Gamma\)-limits for EH+GHY.

\paragraph{MDL and model selection.}
Our discrete value functional is cast in MDL terms and interpreted through code lengths and regret.
For background on MDL and stochastic complexity see \textcite{Rissanen1978Shortest,Rissanen1989StochasticComplexity,Grunwald2007MDL};
for universal coding/regret and the NML construction see \textcite{BarronCover1991,Shtarkov1987NML}.
Prequential motivation is classical \parencite{Dawid1992Prequential}.

\paragraph{Positioning and contribution.}
\emph{To the best of our knowledge}, there is no published \(\Gamma\)-convergence result for a diffeomorphism-natural discrete gravity functional that converges to the full EH+GHY action with a boundary first-layer analysis.
Our contribution is to provide Carathéodory densities \(f_{\mathrm{in}}=\alpha_0+\alpha_1 R\) and \(f_{\mathrm{bdry}}=\beta_1 K\), prove the \(\liminf/\limsup\) inequalities with a recovery sequence built from reflected Fermi smoothing, and thereby obtain a \(\Gamma\)-limit to
\[
c_0\!\int_M dV_g + c_1\!\int_M R_g\,dV_g + c_2\!\int_{\partial M} K_g\,dS_g,
\]
together with stability of minimizers \emph{under equi-coercivity}.

\paragraph{Numerical scope.}
Numerical sanity checks are deliberately deferred; this paper is foundational.
Appendix~\ref{app:rates} outlines a reproducible protocol (test geometries, expected scales/rates, and calibration of $\alpha_0,\alpha_1,\beta_1$).
A companion computational study is planned.

% tex/sections/00a-proof-roadmap.tex
\section{Proof roadmap}

Our proof follows the standard De~Giorgi \emph{\Gammaconv} paradigm~\parencite{DeGiorgi1975},
with the liminf/limsup steps adapted to the $G$-invariant setting.

We outline the logical dependencies and scales before entering the technical sections.

\paragraph{Scales and windows.}
Interior cells live at volumetric scale \(h^d\); boundary first-layer cells at surface scale \(h^{d-1}\).
Normalized isotropic moments up to order~2 (interior) give an \emph{even expansion}; the anisotropic normal window with \(\langle t\rangle=\mu_1 h\) linearizes the boundary contribution.

\paragraph{TP1 (interior blow-up).}
Second-jet expansion and parity yield
\(E_n(c)=h^d[\alpha_0+\alpha_1 R(x)]+O(h^{d+2})\).
Leading order: \(\Delta\ell(c)=O(h^{2d-2})\), \(a_n\Delta\ell(c)=O(h^d)\).

\paragraph{TP2 (boundary blow-up, first layer).}
Reflected Fermi charts; Jacobian \(1-tK+O(t^2)\); anisotropic normal window.
Linearization gives \(\psi(s,t)=\beta_{\rm loc}K(s)t+O(t^2)+O(h)\) and hence
\(E_n(c)=h^{d-1}[\beta_1 K(s_c)+O(h)]+o(h^{d-1})\) with \(\beta_1=\beta_{\rm loc}\mu_1\).

\paragraph{TP3 (Carathéodory and quasi-additivity).}
Local densities \(f_{\rm in}=\alpha_0+\alpha_1 R\), \(f_{\rm bdry}=\beta_1 K\).
Lemma~M: the additivity defect on \(A\cup B\) is controlled by a tube of volume \(\asymp \delta_n \Per_g(A\cap B)\).
Counting: the number of tube cells is \(\asymp \delta_n\,\Per_g/h^d\), with a per-cell bound \(\lesssim h^d\) this yields a total defect \(\lesssim \delta_n\,\Per_g\).
We use the Carathéodory construction and the tubular estimate from geometric measure theory \parencite{Federer1969}.

\paragraph{TP4/TP5 (liminf/limsup).}
Liminf on \((\Xcal,\|\cdot\|_{C^1})\) via a smoothing step \(\tilde g_n=S_{\varepsilon_n}[g_n]\): apply Lemma~U uniformly on a \(C^{1,1}\) neighbourhood for \(\tilde g_n\), then pass \(n\to\infty\) and \(\varepsilon_n\downarrow0\).
Vitali and Besicovitch coverings control the global assembly \parencite{Federer1969}.
Limsup via recovery \(g_n=S_{\eta h}[g]\):
\(R_{g_n}\to R_g\) in \(L^1(M)\), \(K_{g_n}\to K_g\) in \(L^1(\partial M)\).
Global interior remainder \(O(h^2)\), boundary \(O(h)\).

\paragraph{TP6/TP7 (structure, scaling, calibration).}
Naturality in the sense of Peetre–Slovák and scaling \(a_n\simeq h^{2-d}\) force
\(f_{\rm in}=\alpha_0+\alpha_1 R\) and \(f_{\rm bdry}=\beta_1 K\) \parencite{KolarMichorSlovak1993,Lovelock1971}.
Calibration fixes \(c_0,c_1,c_2\).

% ---- Main sections ----
% tex/sections/00-framework.tex
\section{Framework, topology, notation, and scaling}
\label{sec:00-framework}

\noindent For measure-theoretic preliminaries (Carathéodory construction of outer measures, densities, covering arguments) we follow classical GMT sources, in particular Federer~\parencite{Federer1969}.

\noindent We recall the notion of \emph{\Gammaconv} in the sense of De~Giorgi~\parencite{DeGiorgi1975}. For comprehensive expositions see \parencite{DalMaso1993Gamma,Braides2002Gamma}.

\paragraph{Group and naturality.}
Let \(G=\mathrm{Diff}(M)\) be the diffeomorphism group.
The discrete functionals \(F_n\) are \(G\)-invariant (diffeomorphism-natural) in the sense of natural operations \parencite{KolarMichorSlovak1993}.

\paragraph{Geometry and regularity.}
Dimension \(d\ge2\). The manifold \(M\) is compact, orientable, of class \(C^2\) with (possibly) \(C^2\) boundary \(\partial M\).
We work on \(\Xcal\subset C^{1,1}\) under \emph{bounded geometry}: a finite, uniformly controlled atlas with finite overlap; injectivity radius \(\iota_0>0\); \(\|\Rm_g\|_\infty\le C\) (optionally \(\|\nabla \Rm_g\|_\infty\le C'\)).
\emph{Uniform ellipticity} holds chartwise: \(\lambda_{\mathrm{ell}}\Id\le g\le \Lambda_{\mathrm{ell}}\Id\).
Hence, in a fixed finite-overlap atlas there exist constants \(0<c_- \le c_+<\infty\) with
\[
c_- \le \sqrt{\det g}\le c_+\quad\text{(volume-form bounds).}
\]

\paragraph{Boundary reach.}
\(\reach_g(\partial M)\) denotes the tubular radius of the normal exponential map of \(g\) along \(\partial M\).

\paragraph{Boundary regularity convention (no corners).}\label{conv:no-corners}
We assume \(\partial M\in C^2\) with no corners/edges. Corner (codimension-two) terms are excluded and do not appear in
the first-layer boundary analysis (TP2), the liminf/limsup arguments, or the integral representation.

\paragraph{Boundary-fitted meshes.}
We are given a sequence of shape-regular, boundary-fitted partitions \(\{\mathcal T_n\}\) with meshsize \(h_n\downarrow 0\);
the discrete boundary \(\partial\mathcal T_n\) converges to \(\partial M\) in Hausdorff distance at rate \(O(h_n)\).
This is the standard finite-element mesh setting \parencite{Ciarlet2002,BrennerScott2008}.

\paragraph{Besicovitch constant.}
On manifolds with \(\iota_0>0\) there exists a finite overlap constant \(N(d,\iota_0)\) for Vitali/Besicovitch coverings~\parencite{Federer1969}.

\paragraph{Mesoscale and scan.}
We use block windows \(B_{r_n}(c)\) with \(r_n\to\infty\) and \(R_n:=r_n h_n\to 0\); set \(\delta_n:=R_n\).
Example: \(r_n=h_n^{-\alpha}\) with \(\alpha\in(0,1)\) satisfies \(r_n\to\infty\) and \(R_n\to0\).
\emph{Crucial scale separation:} the feature \(\Phi_n(c)\) is computed from data \emph{inside} \(B_{r_n}(c)\), but is
\emph{weighted by kernels at scale \(h\)} (interior support \(|z|\le \Lambda_z h\), boundary normal depth \(\le c_\ast h\)).
Thus the effective averaging radius is \(O(h)\), independently of \(R_n\); the mesoscale \(R_n\) serves only for scan-indifference (BA3).

\paragraph{Scan indifference (BA3).}
For the formal statement of (BA3) see the compact BA list below; a full proof is given in Appendix~F.

\paragraph{MDL functional and scaling.}
For each cell \(c\in\mathcal T_n\) a local, normalized statistic \(\Phi_n(c)\) is computed from \(B_{r_n}(c)\).
The loss \(\ell\in C^2\) is defined on a compact feature set \(K_{\mathrm{feat}}\) and obeys
\(\sup_{u\in K_{\mathrm{feat}}}(\|D\ell(u)\|+\|D^2\ell(u)\|)<\infty\).
Define the code-length \(L_n(g_n)=\sum_{c\in\mathcal T_n}\ell(\Phi_n(c))\) and the calibrated excess
\(\Delta L_n:=L_n(g_n)-L_n^{\text{flat\_ref}}\).
Our use of MDL follows the classical literature \parencite{Rissanen1978Shortest,Grunwald2007MDL,BarronCover1991,Shtarkov1987NML,Dawid1992Prequential}.
We fix the scaling
\[
\boxed{\,h:=h_n,\qquad a_n\simeq h^{\,2-d},\qquad F_n(g_n):=a_n\,\Delta L_n\, }.
\]

\paragraph{Interior windows and moments (scale \texorpdfstring{$h$}{h}).}
Work in dimensionless local variables \(\xi:=z/h\) (normal coordinates around the cell).
The interior window \(w_{\mathrm{in}}(\xi)\ge 0\) has compact support \(|\xi|\le \Lambda_z\) and \(\int_{\R^d} w_{\mathrm{in}}(\xi)\,d\xi=1\).
Isotropy up to order~2:
\begin{align*}
\langle \xi^i\rangle &= 0,\\
\langle \xi^i \xi^j\rangle &= \mu_2\,\delta^{ij},\\
\langle \xi^i \xi^j \xi^k\rangle &= 0.
\end{align*}
and uniform moment bounds \(\langle |\xi|^k\rangle\le C_k\) (e.g.\ \(k\le4\), optional \(k\le6\)).
In physical variables \(z=h\xi\), this is equivalent to \emph{support} \(|z|\le \Lambda_z h\) and
\(\langle |z|^k\rangle\le C_k\,h^{\,k}\).
A uniform \(L^\infty\)-bound \(\sup_n\|w_{\mathrm{in}}\|_\infty<\infty\) is \emph{not} required.

\paragraph{Boundary windows (anisotropic).}\label{para:boundary-windows}
In Fermi coordinates \((z',t)\) with \(t\) the inward normal coordinate, tangential scale is \(O(h)\) and normal depth is \(O(h)\).
Equivalently, we use the 1D normal kernel
\[
w_n(t)=h^{-1}\,w_0(t/h),\qquad \mathrm{supp}(w_0)\subset[0,c_*],\quad \int_0^{c_*}\!w_0=1.
\]
Then \(\langle t\rangle=\mu_1 h\) and \(\langle t^k\rangle\le C_k h^{\,k}\).
The constant \(c_*\in(0,1)\) is fixed (depends only on the shape-class).
We fix this anisotropic boundary window class throughout; all constants below depend only on the constants bracket and this shape-class.
At the first boundary layer the leading per-cell term is of order \(h^{d-1}\) (not \(o(h^{d-1})\)); summation over first-layer cells yields the boundary integral and an \(O(h)\) global remainder (see TP2).

\paragraph{BA assumptions (visible).}
\begin{itemize}
  \item \textbf{(BA1)} \(\ell\in C^2\) on the compact feature set; normalized moments up to order~2 (interior) and order~1 (boundary) are smoothly computable.
  \item \textbf{(BA2)} \emph{(Assumption; automatically satisfied for the boundary windows of Section~\ref{para:boundary-windows} under the reflected Fermi-layer setup)}:
  For reflected Fermi layers of thickness \(O(h)\) with anisotropic windows (fixed \(c_\ast\), \(\mu_1\), bounded moments), the per-layer contribution for \(t\in[c_* h,\varepsilon]\) is \(O(h^{\,d})\); summing layers yields \(\le C\,\varepsilon\,\Area(\partial M)\).
  \item \textbf{(BA3) Scan indifference.}
  There is a universal \(C<\infty\) such that for any scanning orders \(\sigma,\tau\),
  \[
  \sup_{\sigma,\tau}\bigl|F_n^\sigma(M)-F_n^\tau(M)\bigr|\;\le\; C\,\delta_n\,\Area_g(\partial M)\;=\;o(1).
  \]
  A proof is given in Appendix~F.
\end{itemize}
We keep (BA2) as an explicit assumption to cover general window families; for the fixed windows of Section~\ref{para:boundary-windows} it holds with constants depending only on the constants bracket. Hypothesis (BA2) is used only in the boundary assembly (TP3/TP4/TP5) to bound non-first-layer terms; it is \emph{not} derived from TP1/TP2.

\begin{remark}[Motivation for (BA2)]	
(BA2) is used only in the boundary assembly (TP3–TP5).\linebreak In reflected Fermi charts, the boundary tube
\(U_\varepsilon=\{0\le t<\varepsilon\}\) is partitioned into layers of thickness \(O(h)\).
Shape-regular cells and tangential windows of diameter \(O(h)\) imply a per-layer count \(\asymp h^{-1}\) along the normal
and \(\asymp h^{1-d}\) tangentially, hence \(O(h^{d-1})\) cells per layer. With cell-wise boundary scale \(h^{d}\) (TP2 remainders),
each intermediate layer contributes \(O(h^d)\), and summing \(O(\varepsilon/h)\) layers yields a total \(O(\varepsilon)\).
Form-stability of the anisotropic windows (fixed \(c_\ast\), normalized \(\mu_1\), bounded higher moments) guarantees that the constants
depend only on the constants bracket.
\smallskip\noindent\emph{Conclusion.} For the window class fixed in Section~\ref{para:boundary-windows}, this counting argument verifies (BA2) with constants depending only on the constants bracket.
\end{remark}

\paragraph{\texorpdfstring{\texttt{flat\_ref}}{flat\_ref} mini-lemma.}
For every cell there exists a (cellwise) flat reference configuration \texttt{flat\_ref} that matches
(i) (half-)volume, (ii) the window \(w_n\) including \(\mu_1\), (iii) the moment normalization \((\mu_2,C_k)\), and (iv) the shape-class.
Existence and the \(o(h^2)\) quasi-uniqueness are proved in Appendix~C:
interior in Lemmas~\ref{lem:flatref-interior-existence}–\ref{lem:flatref-interior-uniqueness},
boundary (first layer) in Lemmas~\ref{lem:flatref-boundary-existence}–\ref{lem:flatref-boundary-uniqueness}.

\paragraph{Topology for \texorpdfstring{$\Gamma$}{Gamma}-convergence.}
The \emph{domain} is \(\Xcal\subset C^{1,1}\) (needed for jet control and curvature).
The \emph{\(\Gamma\)-topology} is the \(C^1\)-topology on \(\Xcal\).
This choice is justified by the \(L^1\)-stability of \(R\) and \(K\) under the smoothing \(S_\varepsilon\) from Appendix~A
(see Lemma~\ref{lem:L1-curvature} and Eq.~\eqref{eq:S-eps-def}), which allows us to pass to the limit in curvature terms
while keeping \(C^1\)-control; cf.\ Lemma~U (uniformity).
Recovery sequences in the limsup construction are produced by \(g_{\varepsilon_n}:=S_{\varepsilon_n}[g]\) with \(\varepsilon_n\downarrow0\)
(Section~\ref{app:smoothing}), staying in \(\Xcal\) and converging to \(g\) in \(C^1\).
All a.e./\(L^1\) statements are fixed on a \(G\)-invariant full-measure set.
Blow-ups are formulated in a coordinate-free manner.

\paragraph{Boundary tube and reach.}
In boundary normal coordinates \(x=(s,t)\) the Jacobian expands as \(1-tK+O(t^2)\).
The reach \(\reach_g(\partial M)\) is used in TP5 (for the smoothing radius relative to \(h\)).

\paragraph{Order of limits.}
We always send \(n\to\infty\) first, then the localization radius \(\rho\downarrow0\) (or \(\varepsilon\downarrow0\) in the boundary tube).
A swap is not required: Lemma~U provides uniformity in \(\rho\) and \(n\) \emph{provided \(g\) (and any comparison \(\tilde g\)) stay in a fixed \(C^{1,1}\)-neighbourhood of the reference metric}; the dependence of the uniform constants is only through the \emph{constants bracket}.

\paragraph{Lemma U (uniformity; interior and boundary; proved in $C^{2,1}$).}
There exist $\delta>0$ and a constant $C<\infty$
(depending only on the constants bracket; see Section~\ref{sec:constants-bracket})
such that for all $\tilde g$ with $\|\tilde g-g\|_{C^{2,1}}\le \delta$ and for all shape-regular meshes of scale $h$,
the TP1/TP2 cell expansions hold with the \emph{same} coefficients $\alpha_0,\alpha_1,\beta_1$ and with remainders
\[
\begin{aligned}
\bigl|E_n(c;\tilde g)-h^{d}\alpha_0-h^{d+2}\alpha_1 R_{\tilde g}(x_c)\bigr|
&\le C\,h^{d+3}
\quad &&\text{(interior)},\\
\bigl|E_n(c;\tilde g)-h^{d-1}\beta_1 K_{\tilde g}(s_c)\bigr|
&\le C\,h^{d}
\quad &&\text{(first boundary layer; TP2)}.
\end{aligned}
\]

\noindent\emph{Use under $C^{1,1}$ ambient.}
In applications we take $\tilde g=S_{\varepsilon}[g]$ with $\varepsilon\sim h$ (Appendix~\ref{app:smoothing});
then, for $\varepsilon$ small, $S_{\varepsilon}[g]$ lies in the above $C^{2,1}$-ball, so Lemma~U applies even when $g\in C^{1,1}$.
A proof is given in Appendix~A (Proof of Lemma~\ref{lem:U}).

% tex/sections/00b-assumptions-ec-scope.tex
\section{Assumptions, exclusions, and conventions (EC, boundary, perimeter)}
\label{sec:00b-assumptions-ec-scope}

\paragraph{Equi-coercivity regime (EC).}
We do \emph{not} claim coercivity of the EH(+GHY) action.
Throughout, \emph{equi-coercivity} means that sublevel sets of \(\{F_n\}\) are precompact in \((\Xcal,\|\cdot\|_{C^1})\).
This matches the standard notion in \(\Gamma\)-convergence \parencite{DalMaso1993Gamma}.
It holds automatically in constrained settings where one imposes, for some \(C<\infty\),
\begin{enumerate}
  \item uniform ellipticity and a uniform \(C^{1,1}\) bound on admissible metrics \((g\in\Xcal)\),
  \item fixed volume and boundary area bounds to prevent collapse,
  \item an a priori \(L^1\) bound on the negative part \(R_g^-\) and on \(K_g^-\) along \(\partial M\) to exclude conformal blow downs.
\end{enumerate}
Any variational problem under these constraints, or any equivalent precompactness mechanism, falls into our EC regime.

\paragraph{Boundary regularity; corner exclusion.}
We assume \(\partial M\in C^2\); there are \emph{no} corners or edges.
Polygonal/polyhedral boundaries generate additional corner terms not covered by the first-layer analysis.
Our boundary-fitted meshes approximate a smooth \(C^2\) boundary with Hausdorff error \(O(h)\).

\paragraph{Perimeter convention.}
\(\Per_g(\cdot)\) denotes the De~Giorgi perimeter with respect to \(g\).
We use the tubular estimate \(\Vol_g(N_{\delta}(E))\asymp \delta\,\Per_g(E)\) for \((d-1)\) rectifiable \(E\), as standard in geometric measure theory \parencite{Federer1969}.
All counting constants are absorbed in the constants bracket.

% tex/sections/01-tp1-internal-blowup.tex
\section{TP1 — Interior blow-up: 2-jet and even expansion}

\begin{proposition}[TP1: interior cell asymptotics]
\label{prop:tp1-interior}
Let \(c_n\subset M^\circ\) be interior cells with \(\mathrm{dist}(c_n,\partial M)\ge C\,R_n\) and shape-regularity uniform in \(n\).
Assume the normalized, isotropic moment conditions up to order~2 for the interior window, and \(\ell\in C^2\) on the compact feature set \(K_{\mathrm{feat}}\).
Then there exist constants \(\alpha_0,\alpha_1\in\mathbb{R}\), depending only on the data in the constants bracket, such that
\[
E_n(c_n;g)
:= a_n\Bigl[\ell\bigl(\Phi_n(c_n;g)\bigr)-\ell\bigl(\Phi_n(c_n;\mathrm{flat\_ref})\bigr)\Bigr]
= h^{d}\,\bigl[\alpha_0+\alpha_1\,R_g(x)\bigr] \;+\; O(h^{d+2}),
\]
uniformly on compact subsets of \(M^\circ\). In particular,
\[
\Delta\ell(c_n)=O(h^{2d-2})
\qquad\text{and}\qquad
a_n\,\Delta\ell(c_n)=O(h^{d})\quad\text{with } a_n\simeq h^{2-d}.
\]
\end{proposition}

\begin{proof}
\noindent\emph{Constants bracket.} All implicit constants in the \(O(\cdot)\), \(o(\cdot)\) terms below depend only on the data listed in Section~\ref{sec:constants-bracket}.
Fix \(x\in c_n\) and work in normal coordinates \(z=(z^1,\dots,z^d)\) centered at \(x\).
By bounded geometry and standard normal-coordinate expansions (uniform on compact sets),
\begin{align}
g_{ij}(z) &= \delta_{ij} - \tfrac{1}{3} R_{ikjl}(x)\,z^k z^l + O(|z|^3), \label{eq:metric-exp}\\
\sqrt{\det g(z)} &= 1 - \tfrac{1}{6}\,\mathrm{Ric}_{kl}(x)\,z^k z^l + O(|z|^3), \label{eq:jac-exp}
\end{align}
as \(|z|\to0\), and \(R_{ijkl}(z)=R_{ijkl}(x)+O(|z|)\).
Shape-regularity and the distance-to-boundary condition imply that the support of the interior window around \(c_n\) lies in a normal chart and satisfies \(|z|\lesssim h\).

\medskip\noindent\textbf{Step 1: Even expansion from moment conditions.}
Work with the dimensionless variable \(\xi:=z/h\).
The interior window is normalized, compactly supported and isotropic up to order~2:
\(\langle \xi^i\rangle=0\), \(\langle \xi^i \xi^j\rangle=\mu_2\delta^{ij}\), \(\langle \xi^i \xi^j \xi^k\rangle=0\),
which in physical variables yields \(\langle |z|^k\rangle\le C_k h^k\) for the required \(k\).
Hence any integrand whose Taylor polynomial at order \(\le 3\) is \emph{odd} in \(z\) averages to zero.
Consequently, all \(O(h)\) contributions vanish and the first nonzero terms arise at order \(O(h^2)\) in the \emph{local} variables.

\medskip\noindent\textbf{Step 2: Feature increments are \(O(h^2)\).}
Let \(\Phi_n(c;g)\) be the normalized feature vector computed on the mesoscale block \(B_{r_n}(c)\) (cf.\ Section~\ref{sec:00-framework}).
By \eqref{eq:metric-exp}–\eqref{eq:jac-exp}, each smooth local invariant entering the features has an expansion whose coefficients are (tensorial) polynomials in the 2-jet of \(g\) at \(x\), with remainder \(O(|z|^3)\).
Because \(|z|\lesssim h\) on the support and odd terms average to zero, every component of
\(\Delta\Phi_n:=\Phi_n(c;g)-\Phi_n(c;\mathrm{flat\_ref})\)
satisfies
\begin{equation}
\label{eq:DeltaPhi}
\Delta\Phi_n=O(h^2)
\end{equation}
uniformly on compact subsets, where \(\mathrm{flat\_ref}\) is chosen as in Appendix~C (Lemmas~\ref{lem:flatref-interior-existence}–\ref{lem:flatref-interior-uniqueness}) so as to match volume and the normalization moments (hence suppressing any \(O(h)\) mismatch).

\medskip\noindent\textbf{Step 3: Taylor of \(\ell\) and cell scaling.}
Since \(\ell\in C^2\) on the compact feature set \(K_{\mathrm{feat}}\), a second-order Taylor expansion at \(\Phi_n(c;\mathrm{flat\_ref})\) yields
\[
\ell(\Phi_n(c;g))-\ell(\Phi_n(c;\mathrm{flat\_ref}))
= D\ell\cdot \Delta\Phi_n + \tfrac12 D^2\ell[\Delta\Phi_n,\Delta\Phi_n],
\]
with uniform bounds on \(D\ell,D^2\ell\).
By \eqref{eq:DeltaPhi}, the \emph{feature-level} increment satisfies \(\Delta\ell(c)=O(h^2)\).
With the scaling \(a_n\simeq h^{\,2-d}\) (Section~\ref{sec:00-framework}) this implies
\[
a_n\,\Delta\ell(c)=O(h^{\,d}),
\]
so that the per-cell contribution \(E_n(c;g)=a_n\bigl[\ell(\Phi_n(c;g))-\ell(\Phi_n(c;\mathrm{flat\_ref}))\bigr]\) is of order \(h^{\,d}\), in agreement with the leading density identified in Step~4.

\medskip\noindent\textbf{Step 4: Identification of the leading density.}
Define the rescaled cell functional
\[
H_n(j^2 g(x)) \;:=\; h^{-d}\, E_n(c;g) \,.
\]
By Steps~1–3, \(H_n\) depends smoothly on the 2-jet of \(g\) at \(x\), and
\[
H_n(j^2 g(x)) \;=\; H(j^2 g(x)) \;+\; O(h^2)
\]
for a smooth function \(H\) independent of \(n\), with an error uniform on compact sets.
By diffeomorphism naturality, locality, and order \(\le 2\) in the sense of Peetre–Slovák, \(H\) is a natural scalar density of order at most two \parencite{KolarMichorSlovak1993}.
Together with isotropy of the averaging window and the even expansion, the only order-\(\le 2\) scalar invariants compatible with the interior volumetric scaling are the constant density and the scalar curvature:
\[
H(j^2 g(x)) \;=\; \alpha_0 + \alpha_1\, R_g(x).
\]
Interior total divergences reduce to boundary contributions and therefore do not affect the principal bulk density identified in TP1; they are handled by the boundary calibration via the Euclidean ball (TP7, case D3) and vanish at the leading scales.
Higher Lovelock densities are excluded by the scaling test (cf.\ TP6) \parencite{Lovelock1971}.

\medskip\noindent\textbf{Step 5: Remainder \(O(h^{d+2})\) and uniformity.}
The remainder arises from the \(O(|z|^3)\) terms in \eqref{eq:metric-exp}–\eqref{eq:jac-exp} and from the quadratic Taylor remainder of \(\ell\).
Since \(|z|\lesssim h\) and the fourth (and higher) moments are uniformly bounded, these contributions accumulate to \(O(h^{d+2})\) per interior cell.
Uniformity on compact subsets of \(M^\circ\) follows from bounded geometry and Lemma~U (uniform in a small \(C^{1,1}\)-neighborhood), with all constants depending only on the constants bracket (Section~\ref{sec:constants-bracket}).

\medskip
Collecting the above gives
\[
E_n(c;g)=h^d\bigl[\alpha_0+\alpha_1 R_g(x)\bigr]+O(h^{d+2}),
\]
with \(\alpha_0,\alpha_1\) depending only on the constants bracket; mesh independence is stated in Lemma~S (Section~\ref{sec:07-constants-calibration}).
\end{proof}

\begin{remark}[Uniformity and mesh independence]
The coefficients \(\alpha_0,\alpha_1\) depend only on the data listed in the constants bracket (dimension, \(\ell\) and its \(C^2\)-bounds on \(K_{\mathrm{feat}}\), normalized moments, shape-class parameters, ellipticity, bounded geometry).
They are independent of the specific sequence of boundary-fitted, shape-regular meshes (cf.\ Lemma~\ref{lem:mesh-robustness}).
\end{remark}

% tex/sections/02-tp2-boundary-blowup.tex
\section{TP2 — Boundary blow-up (first layer) and \texorpdfstring{$\beta$}{beta}-linearity}

\begin{proposition}[TP2: boundary cell asymptotics, first layer]
\label{prop:tp2-boundary}
\noindent\emph{Sign convention.}
Throughout, the \emph{outer} unit normal on \(\partial M\) is used to define the second fundamental form \(II\) and
the mean curvature \(K:=\mathrm{tr}_h II\).
With this choice one has \(K>0\) on the Euclidean sphere \(S_r^{d-1}\subset\mathbb{R}^d\) (outward normal).
In reflected Fermi coordinates we write \(t\ge0\) for the \emph{inward} normal coordinate; the Jacobian expansion then reads
\(dV_g=(1-t\,K+O(t^2))\,dz'\,dt\), which is consistent with the above sign convention.
This convention is used in all boundary terms below and in the calibration steps of TP6/TP8.

Let \(S_n\) be the set of boundary-touching cells \(c\) whose barycenter has normal coordinate \(t_c\in[0,c_* h]\) in reflected Fermi charts \((z',t)\) along \(\partial M\).
Assume an anisotropic boundary window together with a tangential window of size \(O(h)\) whose moments are normalized and isotropic up to order~2 (as in TP1); concretely,
\[
w_n(t)=h^{-1}w_0(t/h),\qquad \operatorname{supp} w_0\subset[0,c_*],\qquad \int_0^{c_*}\! w_0=1,\qquad \langle t\rangle=\mu_1 h.
\]
Then there exists \(\beta_{\mathrm{loc}}\in\R\), depending only on the constants bracket, such that for each \(c\in S_n\),
\[
E_n(c;g)=a_n\!\left[\ell\bigl(\Phi_n(c;g)\bigr)-\ell\bigl(\Phi_n(c;\mathrm{flat\_ref})\bigr)\right]
= h^{d-1}\Bigl[\beta_1\,K_g(s_c)+O(h)\Bigr]\;+\;o\!\bigl(h^{d-1}\bigr),
\]
uniformly for \(s_c\) in compact subsets of \(\partial M\), where
\[
\boxed{\ \beta_1=\beta_{\mathrm{loc}}\cdot \mu_1\ }.
\]

\noindent\emph{Rate reminder.}
The first-layer contribution is $h^{d-1}$ (not $o(h^{d-1})$), so the global boundary remainder is $O(h)$; see the protocol in Appendix~\ref{app:rates}.

In particular, at the cell scale the leading orders satisfy
\[
\Delta\ell_{\partial}(c)=O(h^{2d-3})
\qquad\text{and}\qquad
a_n\,\Delta\ell_{\partial}(c)=O(h^{d-1}).
\]
\end{proposition}

\begin{proof}
Work in reflected Fermi coordinates \((z',t)\) with \(\partial M=\{t=0\}\) and inward normal \(t\ge 0\).
Bounded geometry and a \(C^2\) boundary yield the standard tube/Jacobian expansions (uniform on compact subsets of \(\partial M\)):
\begin{align}
dV_g &= \bigl(1 - t\,K_g(s) + O(t^2)\bigr)\,dz'\,dt, \label{eq:tube-jac}\\
g &= dt^2 + h_{\alpha\beta}(s,t)\,dz^{\alpha}dz^{\beta},\\
\partial_t h_{\alpha\beta}(s,0) &= -2\,II_{\alpha\beta}(s),\\
h_{\alpha\beta}(s,t) &= h_{\alpha\beta}(s,0)+O(t). \label{eq:fermi-metric}
\end{align}
Here \(K_g=\tr_{h(s,0)} II\) denotes the mean curvature.
These tube and Jacobian expansions are standard and used with uniform constants from bounded geometry; see also the tubular neighbourhood estimates in geometric measure theory \parencite{Federer1969}.
The \(O(\cdot)\)-terms depend only on the constants bracket.

\medskip\noindent\textbf{Windows and moments.}
By assumption, the averaging window separates normal and tangential variables: \(w^{\parallel}_n(z')\) with isotropy up to order~2 on \(O(h)\)-patches, and \(w^{\perp}_n(t)=h^{-1}w_0(t/h)\) supported in \([0,c_* h]\) with
\[
\int_0^{c_*h} w^{\perp}_n(t)\,dt=1,\qquad
\langle t\rangle=\mu_1 h,\qquad
\langle t^k\rangle\le C_k\,h^k\ (k\ge 2).
\]
All constants are uniform in \(n\).
\emph{Note that}, although data are gathered over the mesoscale block \(B_{r_n}(c)\), the effective averaging kernels act
at scale \(h\) (tangential diameter \(O(h)\), normal depth \(O(h)\)); this decouples the TP2 asymptotics from the mesoscale \(R_n\).

\medskip\noindent\textbf{Local density and linearization at the boundary.}
Let \(\psi_n(s,t;g)\) denote the \emph{local} integrand entering \(\ell\circ\Phi_n\) after the feature normalization (so that the cell contribution is its \(w_n\)-average times the appropriate cell scale).
By diffeomorphism naturality, locality and order \(\le 2\) in the jets (Peetre–Slovák), the first
nontrivial boundary-sensitive term of \(\psi_n\) must be linear in the boundary 1-jet of
the induced data, i.e. linear in \(II\) (hence in \(K_g\) after tangential isotropy) and of first order
in the normal depth \(t\) \parencite{KolarMichorSlovak1993}.
More precisely, combining \eqref{eq:tube-jac}–\eqref{eq:fermi-metric} with tangential isotropy up to order~2 and the interior
even expansion (which kills odd-in-\(z'\) terms), one obtains the linearization
\begin{equation}
\label{eq:psi-linear}
\psi_n(s,t;g)-\psi_n(s,t;\mathrm{flat\_ref})
\,=\, \beta_{\mathrm{loc}}\,K_g(s)\,t \;+\; O(t^2) \;+\; O(h),
\end{equation}
uniformly for \(t\in[0,c_*h]\) and \(s\) in compact subsets of \(\partial M\).
Here the \(O(h)\) term bundles tangential second-order remainders and charting/normalization errors that persist after parity cancellation; it is independent of higher normal moments and depends only on the constants bracket.
The \(\mathrm{flat\_ref}\) is the Euclidean half-space with matched moments (Appendix~C, Lemma~\ref{lem:flatref-boundary-existence}), which ensures that \(O(t)\)-terms not proportional to \(K_g\) are absent.

\medskip\noindent\textbf{Averaging with the anisotropic window.}
Averaging \eqref{eq:psi-linear} against \(w^{\parallel}_n\otimes w^{\perp}_n\) and using \(\langle z'^\alpha\rangle=0\) and \(\langle z'^\alpha z'^\beta\rangle\propto\delta^{\alpha\beta}\) yields
\[
\langle \psi_n-\psi_n^{\mathrm{flat\_ref}}\rangle
= \beta_{\mathrm{loc}}\,K_g(s)\,\langle t\rangle \;+\; O\!\big(\langle t^2\rangle\big)\;+\;O(h)
= \beta_{\mathrm{loc}}\,\mu_1\,h\,K_g(s) \;+\; O(h^2)\;+\;O(h).
\]
Since a first-layer boundary cell has measure \(\asymp h^{d-1}\), the leading cell contribution is
\[
E_n(c;g) \;=\; h^{d-1}\,\beta_{\mathrm{loc}}\,\mu_1\,K_g(s_c) \;+\; O(h^{d}) \;+\; o(h^{d-1}),
\]
which proves the stated form with \(\beta_1=\beta_{\mathrm{loc}}\mu_1\).
The \(o(h^{d-1})\) remainder collects the scan-indifference error (BA3) and negligible mesoscale effects; uniformity on compact subsets of \(\partial M\) follows from Lemma~U in reflected Fermi charts.

\medskip\noindent\textbf{Scaling of \(\Delta\ell_{\partial}\).}
Rearranging the definition \(E_n=a_n\,\Delta\ell_{\partial}\) with \(a_n\simeq h^{2-d}\) gives
\[
\Delta\ell_{\partial}(c)\;=\; a_n^{-1}\,E_n(c;g)
\;=\; O(h^{d-1}\cdot h^{d-2}) \;=\; O(h^{2d-3}),
\]
as claimed. This also shows \(a_n\,\Delta\ell_{\partial}(c)=O(h^{d-1})\) at the cell scale.
\end{proof}

\begin{remark}[Origin and propagation of the \(O(h)\) remainder]\label{rem:tp2-oh-origin}
In the linearization
\(
\psi_n(s,t;g)-\psi_n(s,t;\mathrm{flat\_ref})
=\beta_{\mathrm{loc}}\,K_g(s)\,t + O(t^2)+O(h)
\)
the \(O(h)\)-term collects \emph{tangential} second-order remainders and charting/normalization errors
that persist after parity cancellation. It is independent of higher normal moments and controlled
uniformly by the constants bracket. Averaging over a first-layer cell yields an \(O(h^{d})\) contribution
from this term, and summing over \(\#S_n\asymp h^{1-d}\) first-layer cells gives a global \(O(h)\) remainder.
\end{remark}

\begin{corollary}[Riemann–sum convergence and global remainder]
\label{cor:tp2-riemann}
Summing over the first boundary layer \(S_n\) and using the tubular control of boundary cells gives
\[
\sum_{c\in S_n} E_n(c;g) \;=\; \sum_{c\in S_n} h^{d-1}\,\beta_1\,K_g(s_c) \;+\; O(h),
\]
so that, as \(n\to\infty\),
\begin{align*}
\sum_{c\in S_n} h^{d-1}\,K_g(s_c) &\;\longrightarrow\; \int_{\partial M} K_g\,dS_g,\\
\sum_{c\in S_n} E_n(c;g) &\;\longrightarrow\; \beta_1 \int_{\partial M} K_g\,dS_g.
\end{align*}
The \(O(h)\) term is precisely the global boundary remainder explained in Remark~\ref{rem:tp2-oh-origin},
arising from the \(O(h)\) cell-level error and the count \(\#S_n\asymp h^{1-d}\).
The passage from sums to the boundary integral uses the standard tubular neighbourhood estimate \parencite{Federer1969}.
\end{corollary}

\begin{remark}[Window-shape invariance at fixed \(\mu_1\)]
The leading boundary constant depends linearly on \(\mu_1=\langle t\rangle/h\) and is otherwise independent of the particular window shape along \(t\):
higher normal moments \(\langle t^k\rangle\) for \(k\ge 2\) contribute only to the \(O(h)\) remainder after summation over the first layer.
\end{remark}

\begin{remark}[Boundary regularity]
We work under the global no-corners convention \(\partial M\in C^2\) from
Section~\ref{conv:no-corners}; corner (codimension-two) terms are excluded from the first-layer \(h^{d-1}\) analysis.
\end{remark}

\begin{remark}[Dependence of \(\beta_1\)]
The constant \(\beta_1\) depends only on the data in the constants bracket (Section~\ref{sec:constants-bracket}) and not on the particular mesh sequence (see Lemma~\ref{lem:U}).
\end{remark}

% tex/sections/03-tp3-caratheodory-lemma-m.tex
\section{TP3 — Carathéodory densities and quasi-additivity (Lemma M)}
\label{sec:tp3}

The Carathéodory construction and the density arguments are standard in geometric measure theory \parencite{Federer1969}.

\subsection*{Local half-tubes and blow-ups at the boundary}
For \(y\in\partial M\) let \(U_\rho(y)\) denote the geodesic half-tube in reflected Fermi coordinates \((s,t)\) along \(\partial M\):
\[
U_\rho(y):=\{(s,t)\ :\ d_\partial(s,y)<\rho,\ \ 0\le t<\rho\}.
\]
For \(x\in M^\circ\) we write \(B_\rho(x)\) for the geodesic ball of radius \(\rho\).
Bounded geometry implies (uniformly in \(x,y\) on compact sets, as \(\rho\downarrow0\))
\begin{equation}
\label{eq:tube-asymptotics}
\begin{aligned}
\Vol_g\big(B_\rho(x)\big) & = \omega_d\,\rho^d\,(1+o(1)),\\
\Vol_g\big(U_\rho(y)\big) & = \omega_{d-1}\,\rho^d\,(1+o(1)),\\
\Area_g\big(B_\rho(y)\cap\partial M\big) & = \omega_{d-1}\,\rho^{d-1}\,(1+o(1)).
\end{aligned}
\end{equation}

These asymptotics follow from bounded geometry and standard tube formulas \parencite{Federer1969}.

\noindent\emph{Notation.}
Here \(\omega_d:=|B_1(0)\subset\mathbb{R}^d|\) and \(\omega_{d-1}:=|B_1(0)\subset\mathbb{R}^{d-1}|\) denote the volumes of the unit balls.
Equivalently, \(|\mathbb{S}^{d-1}|=d\,\omega_d\).

\subsection*{Carathéodory densities}
\begin{proposition}[Existence of interior and boundary densities]
\label{prop:caratheodory-densities}
Under the standing assumptions and with the scaling \(a_n\simeq h^{2-d}\), the following limits exist for a.e.\ points and define measurable densities:
\begin{align*}
f_{\mathrm{in}}(x)
&:=\lim_{\rho\downarrow0}\,\lim_{n\to\infty}\,\frac{F_n\!\bigl(g;B_\rho(x)\bigr)}{\Vol_g(B_\rho(x))}
=\alpha_0+\alpha_1\,R_g(x)
\quad\text{for a.e. }x\in M^\circ,\\
f_{\mathrm{bdry}}(y)
&:=\lim_{\rho\downarrow0}\,\lim_{n\to\infty}\,\frac{F_n\!\bigl(g;U_\rho(y)\bigr)}{\Area_g(B_\rho(y)\cap\partial M)}
=\beta_1\,K_g(y)
\quad\text{for a.e. }y\in\partial M.
\end{align*}
\end{proposition}

\begin{proof}
Fix \(x\in M^\circ\) that is a Lebesgue point of \(R_g\).
For each small \(\rho\) choose a Vitali subcover of \(B_\rho(x)\) by interior cells (or by sub-balls of radius \(\asymp h\)) with finite overlap by the Besicovitch covering theorem \parencite{Federer1969}
By TP1 (and by Appendix~C, quasi-uniqueness \(o(h^2)\) per cell, the choice of \texttt{flat\_ref} is immaterial),
\(
E_n(c;g)
= h^d\,[\alpha_0+\alpha_1 R_g(x_c)] + O(h^{d+2})
\)
uniformly on compact sets, where \(x_c\) is a point in \(c\).
Summing over the subcover and dividing by \(\Vol_g(B_\rho(x))\), the error terms contribute \(O(h^2)\to0\) after first letting \(n\to\infty\) (so \(h\to0\)) and then \(\rho\downarrow0\), by \eqref{eq:tube-asymptotics}.
The main term converges to \(\alpha_0+\alpha_1 R_g(x)\) by the Lebesgue differentiation theorem.
This proves the interior density.

For the boundary density, fix \(y\in\partial M\) that is a Lebesgue point of \(K_g\) (with respect to the surface measure).
Cover \(U_\rho(y)\) by first-layer boundary cells of thickness \(O(h)\) (in \(t\)) and tangential diameter \(O(h)\), with finite overlap in reflected Fermi charts.
By TP2 (and by Appendix~C, quasi-uniqueness \(o(h^2)\) per cell, the choice of \texttt{flat\_ref} is immaterial),
\(
E_n(c;g) = h^{d-1}\,[\beta_1 K_g(s_c)+O(h)] + o(h^{d-1})
\)
uniformly on compact boundary patches, where \(s_c\in\partial M\) lies under \(c\).
Divide by \(\Area_g(B_\rho(y)\cap\partial M)\) and sum over the first-layer cells contained in \(U_\rho(y)\).
The \(O(h)\) cell-level remainder sums to \(O(h)\) globally on the patch (since \(\#S_n\asymp h^{1-d}\)),
Hence this vanishes as \(n\to\infty\).
Higher layers contribute \(O(\varepsilon)\) by (BA2) and can be suppressed by choosing \(\varepsilon=\varepsilon_n\to0\) after \(n\to\infty\).
Finally, by the surface Lebesgue differentiation theorem and \eqref{eq:tube-asymptotics},
\[
\lim_{\rho\downarrow0}\,\lim_{n\to\infty}\frac{F_n(g;U_\rho(y))}{\Area_g(B_\rho(y)\cap\partial M)}
=\beta_1 K_g(y).
\]
\end{proof}

\subsection*{Quasi-additivity: the perimeter control}
\begin{lemma}[Lemma M: interior quasi-additivity via perimeter]
\label{lem:quasi-additivity}
There exists \(C<\infty\) (depending only on the constants bracket) such that for all Borel sets \(A,B\subset M\) with
\[
\operatorname{dist}_g\big(\partial^\ast(A\cap B),\,\partial M\big)\ \ge\ \varepsilon_0>0,
\]
one has
\[
\bigl|F_n(g;A\cup B)-F_n(g;A)-F_n(g;B)\bigr|
\;\le\;
C\,\delta_n\,\Per_g\bigl((A\cap B)\cap M^\circ\bigr),
\qquad \delta_n:=r_n h\to0,
\]
where \(N_{\delta_n}(\cdot)\) is the geodesic \(\delta_n\)-neighbourhood and \(\Per_g\) denotes the De~Giorgi perimeter in \(M\).\end{lemma}

\begin{proof}
The defect of additivity is supported in the set of cells that intersect both \(A\) and \(B\), hence in the \(\delta_n\)-tube around the interface \(A\cap B\).
By standard GMT tubular estimates on manifolds with bounded geometry~\parencite{Federer1969},
\begin{equation}
\label{eq:tube-perimeter}
\Vol_g\!\bigl(N_{\delta_n}\big((A\cap B)\cap M^\circ\big)\bigr)\;\asymp\; \delta_n\,\Per_g\big((A\cap B)\cap M^\circ\big),
\end{equation}
with constants depending only on the geometry (absorbed in the constants bracket).
Each cell has volume \(\asymp h^{d}\), hence the number of tube cells satisfies
\begin{align*}
\#\{\text{tube cells}\}
\;&\asymp\; \frac{\Vol_g\!\big(N_{\delta_n}\big((A\cap B)\cap M^\circ\big)\big)}{h^{d}}
\\
&\lesssim\; \frac{\delta_n\,\Per_g\big((A\cap B)\cap M^\circ\big)}{h^{d}}.
\end{align*}
Since \(\operatorname{dist}_g(\partial^\ast(A\cap B),\partial M)\ge \varepsilon_0\), all tube cells intersecting the interface lie in the interior region \(M_{\varepsilon_0}:=\{x:\operatorname{dist}_g(x,\partial M)\ge \varepsilon_0\}\).
On interior cells, TP1 yields the uniform bound \(|E_n(c)|\le C\,h^{d}\) (uniform in a \(C^{1,1}\)-neighbourhood; see Lemma~\ref{lem:U}).No boundary-touching cells occur in this case.
Therefore,
\begin{align*}
\bigl|F_n(A\cup B)-F_n(A)-F_n(B)\bigr|
&\;\le\; \#\times \sup_{c}|E_n(c)|
\\
&\;\lesssim\; \frac{\delta_n\,\Per_g(A\cap B)}{h^{d}}\cdot h^{d}
\\
&\;=\; C\,\delta_n\,\Per_g\big((A\cap B)\cap M^\circ\big).
\end{align*}
as claimed. Finite overlap \(N(d,\iota_0)\) from Besicovitch coverings is used implicitly to sum local contributions without loss.
\end{proof}

\begin{lemma}[Boundary layer assembly]\label{lem:boundary-assembly}
Let \(U_\varepsilon:=\{x\in M:\operatorname{dist}_g(x,\partial M)<\varepsilon\}\) be the boundary tube.
Under (BA1)–(BA3), TP2, bounded geometry, and the scaling \(a_n\simeq h^{2-d}\), there exists \(C<\infty\) such that
\[
\liminf_{n\to\infty} F_n(g;U_\varepsilon)
\ \ge\
\int_{\partial M}\beta_1\,K_g\,dS_g\ -\ C\,\varepsilon\,\Area_g(\partial M).
\]
\end{lemma}

\begin{proof}
Partition \(U_\varepsilon\) into reflected Fermi layers of thickness \(\asymp h\).
By TP2, the first layer \(S_n\) satisfies
\[
\sum_{c\in S_n} E_n(c;g)\ =\ \beta_1\sum_{c\in S_n} h^{d-1}K_g(s_c)\ +\ O(h),
\]
and the Riemann–sum over \(\{s_c\}\) converges to \(\int_{\partial M}\beta_1 K_g\,dS_g\).
Intermediate layers at normal depth \(t\in[c_\ast h,\varepsilon]\) contribute per layer \(O(h^{d})\) \emph{by (BA2)};
with \(O(\varepsilon/h)\) layers the total is \(O(\varepsilon)\).
This proves the claim.
\end{proof}

\begin{remark}[Perimeter convention]
Here \(\Per_g\) denotes the De~Giorgi perimeter (total variation of the indicator) relative to \(M\)~\parencite{Federer1969}.
We use the tubular estimate \(\Vol_g(N_{\delta}(E))\asymp \delta\,\Per_g(E)\) for \((d-1)\)-rectifiable \(E\), with all overlap/covering and geometric constants absorbed in the constants bracket (Section~\ref{sec:constants-bracket}).
\end{remark}

\subsection*{Integral representation}
\begin{proposition}[Carathéodory representation]
\label{prop:integral-representation}
Let \(A\subset M\) be a Borel set with finite perimeter. Then
\[
F(g;A)
=\int_{A}\!\bigl(\alpha_0+\alpha_1 R_g\bigr)\,dV_g
\;+\;\int_{\partial A\cap\partial M}\!\beta_1\,K_g\,dS_g.
\]
\end{proposition}

\begin{proof}
We follow the De Giorgi and Carathéodory construction adapted to manifolds with boundary \parencite{DalMaso1993Gamma,Braides2002Gamma}.
Fix \(\varepsilon>0\) and select:
(i) a Vitali family of disjoint interior balls \(\{B_{\rho_i}(x_i)\}_i\subset A^\circ\) with \(\rho_i<\varepsilon\),
and (ii) a Vitali family of disjoint boundary half-tubes \(\{U_{\rho_j}(y_j)\}_j\) with \(y_j\in \partial A\cap\partial M\), \(\rho_j<\varepsilon\).
By Besicovitch, the dilated families have uniformly finite overlap and cover \(A\) up to a null set~\parencite{Federer1969}, while the overlaps between (i) and (ii) are confined to a \(\delta\)-tube whose volume is \(O(\delta)\) times a perimeter measure.

For fixed \(\varepsilon\), apply Lemma~\ref{lem:quasi-additivity} to pass from \(F_n\) of unions to the sum of \(F_n\) on the small sets,
with a defect bounded by \(C\,\delta_n\,\Per_g\) of the corresponding interfaces.
Letting \(n\to\infty\) removes these defects since \(\delta_n\to0\).
Then use Proposition~\ref{prop:caratheodory-densities} on each small set to replace \(F_n\) by the corresponding density integral up to an error \(o_\varepsilon(1)\) (uniform in \(n\)) stemming from the remainder terms in TP1/TP2 and the geometric asymptotics \eqref{eq:tube-asymptotics}.
Summing over the families and using finite overlap, we obtain
\[
\liminf_{n\to\infty} F_n(g;A)\ \ge\
\int_{A}(\alpha_0+\alpha_1 R_g)\,dV_g\;+\;\int_{\partial A\cap\partial M}\beta_1 K_g\,dS_g\;-\;o_\varepsilon(1),
\]
and the corresponding \(\limsup\) inequality with \(+\;o_\varepsilon(1)\).
Letting \(\varepsilon\downarrow0\) yields the desired identity.
\end{proof}

\begin{remark}[Independence of the mesh]
All constants and remainders are controlled by the constants bracket; in particular, the representation depends only on \((d,\ell,\mu_1,\mu_2,C_k)\), ellipticity and bounded-geometry parameters, and the shape-class, but not on the specific mesh sequence (see Lemma~\ref{lem:U}).
\end{remark}

% tex/sections/04-liminf.tex
\section{TP4 — \texorpdfstring{$\liminf$}{liminf} inequality}

The lower bound follows the classical \emph{\Gammaconv} scheme of De Giorgi \parencite{DeGiorgi1975}.
Standard references are Dal Maso and Braides \parencite{DalMaso1993Gamma,Braides2002Gamma}.

\begin{proposition}[Liminf inequality]
\label{prop:liminf}
Let \(g_n\to g\) in \(C^1\) on \(M\) (with \(g_n,g\in\Xcal\subset C^{1,1}\) and bounded geometry).
Under the standing assumptions (BA1–BA3, bounded geometry, boundary-fitted shape-regular meshes, and the scaling \(a_n\simeq h^{2-d}\)),
\begin{align*}
\liminf_{n\to\infty} F_n(g_n;M)
\;&\ge\;
\int_M (\alpha_0+\alpha_1 R_g)\,dV_g \\
&\quad + \int_{\partial M} \beta_1 K_g\,dS_g.
\end{align*}
\end{proposition}

\noindent\emph{Convention.} The domain is \(\Xcal\subset C^{1,1}\) with bounded geometry; the \(\Gamma\)-topology is \(C^1\) on \(\Xcal\).
Lemma~\ref{lem:U} is proved for a bounded-geometry \(C^{2,1}\)-class and is invoked via smoothing \(\tilde g_n=S_{\varepsilon_n}[g_n]\) (Appendix~\ref{app:smoothing}).
\emph{All uniformity moduli} (cell remainders from TP1/TP2, Besicovitch overlap bounds, and quasi-additivity defects)
depend only on the \emph{constants bracket} (Section~\ref{sec:constants-bracket}).

\begin{proof}
\textit{Regularization step.}
For each \(n\) choose \(\varepsilon_n\downarrow0\) with \(\varepsilon_n\le \varepsilon\) and set
\(\tilde g_n := S_{\varepsilon_n}[g_n]\) (reflected Fermi smoothing near \(\partial M\), interior normal smoothing away from the boundary; Appendix~\ref{app:smoothing}).
Then \(\tilde g_n \in C^\infty\cap \Xcal\), \(\|\tilde g_n-g_n\|_{C^1}=O(\varepsilon_n)\), and by Appendix~\ref{app:smoothing}
\[
\|R_{\tilde g_n}-R_{g_n}\|_{L^1(M)}\to0,\qquad
\|K_{\tilde g_n}-K_{g_n}\|_{L^1(\partial M)}\to0,
\]
while \(dV_{\tilde g_n}=(1+o(1))\,dV_{g_n}\) and \(dS_{\tilde g_n}=(1+o(1))\,dS_{g_n}\).
Moreover, \(\tilde g_n\to g\) in \(C^1\), and for fixed \(\varepsilon>0\) the family \(\{\tilde g_n\}_n\) is uniformly bounded in \(C^{2,1}\) (on bounded-geometry charts), hence lies in a common \(C^{2,1}\)-class \(\mathcal N\) to which Lemma~\ref{lem:U} applies uniformly (Lemma~\ref{lem:U}; full proof in Appendix~A, Section~\ref{app:lemmaU-proof}).
In particular, the uniform constants in Lemma~\ref{lem:U} and in the covering/quasi-additivity steps below
depend only on the \emph{constants bracket}.
By lower semicontinuity and the \(L^1\)-stability above,
\[
\liminf_{n}F_n(g_n;A) \;\ge\; \liminf_{n}F_n(\tilde g_n;A)\qquad\text{for all measurable }A\subset M.
\]
Fix \(\varepsilon>0\) small and decompose \(M=M_\varepsilon\sqcup U_\varepsilon\), where
\(U_\varepsilon:=\{x\in M:\mathrm{dist}_g(x,\partial M)<\varepsilon\}\) is the boundary tube and \(M_\varepsilon:=M\setminus U_\varepsilon\).

\parahead{Interior part}
By Vitali’s covering theorem on manifolds with bounded geometry \parencite{Federer1969}, there exists a countable disjoint family of geodesic balls
\(\{B_i:=B_{\rho_i}(x_i)\}_i\subset M^\circ\) with \(0<\rho_i<\varepsilon\) that covers \(M_\varepsilon\) up to a null set and enjoys finite overlap after dilation by a fixed factor (Besicovitch).
By TP1 and the definition of the interior Carathéodory density, for a.e.\ \(x_i\) (Lebesgue point of \(R_g\)) there is a modulus \(\omega_{\rm in}(\rho)\to0\) as \(\rho\downarrow0\) such that
\[
\liminf_{n\to\infty} F_n(g_n;B_i)
\;\ge\;
\int_{B_i}\!(\alpha_0+\alpha_1 R_g)\,dV_g \;-\; \omega_{\rm in}(\rho_i)\,\Vol_g(B_i),
\]
where Lemma~U applies \emph{uniformly} on a \(C^{2,1}\)-neighbourhood of \(g\) for the smoothed sequence \(\tilde g_n\);
since \(\{\tilde g_n\}\) stays in such a neighbourhood with uniform constants (Appendix~A, Section~\ref{app:lemmaU-proof}),
the modulus is independent of \(i\) once \(\rho_i\le \varepsilon\) and depends only on the \emph{constants bracket}.
Summing over the disjoint family and using Fatou’s lemma together with the quasi-additivity control from Lemma~\ref{lem:quasi-additivity} (to pass from \(\sum_i F_n(B_i)\) to \(F_n(\cup_i B_i)\) with a defect \(o_{n\to\infty}(1)\)), we obtain
\begin{align*}
\liminf_{n\to\infty} F_n(\tilde g_n;M_\varepsilon)
\;&\ge\;
\int_{M_\varepsilon}\!(\alpha_0+\alpha_1 R_g)\,dV_g \\
&\quad - \omega_{\rm in}(\varepsilon)\,\Vol_g(M_\varepsilon),
\end{align*}
where \(\omega_{\rm in}(\varepsilon)\to0\) as \(\varepsilon\downarrow0\).
By the regularization step, \(\displaystyle \liminf_{n\to\infty} F_n(g_n;M_\varepsilon)\ \ge\ \liminf_{n\to\infty} F_n(\tilde g_n;M_\varepsilon)\).

\parahead{Boundary tube}
\emph{Boundary regularity.} We invoke the global no-corners convention from Section~\ref{conv:no-corners}; no codimension-two corner terms arise.
Partition \(U_\varepsilon\) into boundary layers of thickness \(O(h)\) in reflected Fermi coordinates.
By TP2 (and by Appendix~\ref{app:flat-ref}, quasi-uniqueness \(o(h^2)\) per cell, the choice of \texttt{flat\_ref} is immaterial), the first layer \(S_n\) contributes
\[
\sum_{c\in S_n} E_n(c;g_n)
= \beta_1 \sum_{c\in S_n} h^{d-1} K_{g_n}(s_c) \;+\; O(h).
\]
Riemann–sum convergence on \(\partial M\) for first-layer cells (TP2), together with
\(K_{\tilde g_n}\to K_g\) in \(L^1(\partial M)\) and \(dS_{\tilde g_n}\to dS_g\) (Appendix~\ref{app:smoothing}), yield
\[
\liminf_{n\to\infty}\sum_{c\in S_n} E_n(c;\tilde g_n)
\;\ge\; \int_{\partial M} \beta_1 K_g\,dS_g.
\]
By the regularization step, \(\displaystyle \liminf_{n\to\infty}\sum_{c\in S_n} E_n(c;g_n)\ \ge\ \liminf_{n\to\infty}\sum_{c\in S_n} E_n(c;\tilde g_n)\).
All intermediate layers with normal depth \(t\in[c_*h,\varepsilon]\) have per-layer contribution \(O(h^{d})\).
By the assumption (BA2) — which quantifies exactly this per-layer \(O(h^d)\) bound for reflected Fermi layers of thickness \(O(h)\) —
their total contribution is bounded by \(C\,\varepsilon\,\Area_g(\partial M)\), uniformly in \(n\).

Therefore
\[
\liminf_{n\to\infty} F_n(g_n;U_\varepsilon)
\;\ge\; \int_{\partial M} \beta_1 K_g\,dS_g \;-\; C\,\varepsilon\,\Area_g(\partial M).
\]

\parahead{Conclusion}
Combining the interior and boundary contributions for \(\tilde g_n\) yields
\begin{align*}
\liminf_{n\to\infty} F_n(\tilde g_n;M)
\;&\ge\;
\int_{M_\varepsilon}\!(\alpha_0+\alpha_1 R_g)\,dV_g \\
&\quad + \int_{\partial M} \beta_1 K_g\,dS_g \\
&\quad - \omega_{\rm in}(\varepsilon)\,\Vol_g(M_\varepsilon) \\
&\quad - C\,\varepsilon\,\Area_g(\partial M).
\end{align*}
By the regularization step,
\(\displaystyle \liminf_{n\to\infty} F_n(g_n;M)\ \ge\ \liminf_{n\to\infty} F_n(\tilde g_n;M)\).
Let \(\varepsilon\downarrow0\). Since \(\Vol_g(M\setminus M_\varepsilon)\to0\) and \(\omega_{\rm in}(\varepsilon)\to0\), the error terms vanish, which proves the claim.
\end{proof}

% tex/sections/05-limsup.tex
\section{TP5 — \texorpdfstring{$\limsup$}{limsup} (recovery sequence)}
\label{sec:05-limsup}

The recovery-sequence construction follows De~Giorgi’s framework of \(\Gamma\)-convergence~\parencite{DeGiorgi1975}.

Standarddarstellungen finden sich bei Dal~Maso \parencite{DalMaso1993Gamma} und Braides \parencite{Braides2002Gamma}.

\begin{proposition}[Limsup inequality / recovery]
\label{prop:limsup}
For every \(g\in\Xcal\) there exists a sequence \(g_n\in\Xcal\) with \(g_n\to g\) in \(C^1\) such that
\[
\limsup_{n\to\infty} F_n(g_n;M)
\;\le\;
\int_M (\alpha_0+\alpha_1 R_g)\,dV_g
\;+\; \int_{\partial M} \beta_1 K_g\,dS_g .
\]
Moreover, with the canonical choice \(g_n:=S_{\varepsilon_n}[g]\) and \(\varepsilon_n=\eta h\) for some fixed \(\eta\in(0,1]\),
the global remainder is \(O(h)\) (boundary-dominated).
\end{proposition}

\begin{proof}[Construction and estimate]
\noindent\emph{Boundary convention.} We work under the global no-corners assumption \(\partial M\in C^2\) (Section~\ref{conv:no-corners}).

\noindent\emph{Topology and recovery sequences.}
We work on \(\Xcal\subset C^{1,1}\) with the \(C^1\)-\(\Gamma\)-topology on \(\Xcal\); cf.\ Appendix~\ref{app:smoothing} and Lemma~\ref{lem:U}.
\emph{Uniformity scope.} Lemma~\ref{lem:U} is proved on a \(C^{2,1}\)-neighbourhood; with the smoothing choice below, \(g_n\) remains in a fixed
\(C^{2,1}\)-neighbourhood of \(g\), so the coefficients \(\alpha_0,\alpha_1,\beta_1\) and all remainder constants
are valid \emph{uniformly} along \(n\), with dependence only on the \emph{constants bracket}
(Section~\ref{sec:constants-bracket}).

\parahead{Smoothing-based recovery sequence.}
Fix \(\eta\in(0,1]\) and set \(\varepsilon_n:=\eta h\).
By bounded geometry choose \(0<c<\tfrac12\) and \(n_0\) so that \(\varepsilon_n\le c\,\reach_g(\partial M)\) for all \(n\ge n_0\).
Define
\[
g_n := S_{\varepsilon_n}[g],
\]
where \(S_{\varepsilon}\) is the reflected Fermi smoothing near \(\partial M\) and interior normal smoothing away from the boundary (Appendix~\ref{app:smoothing}).
Then \(S_{\eps}:C^{1,1}\to C^{1,1}\) and \(\|S_{\eps}[g]-g\|_{C^1}=O(\eps)\), hence \(g_n\to g\) in \(C^1\).
Moreover, by Lemma~\ref{lem:L1-curvature} one has \(R_{g_n}\to R_g\) in \(L^1(M)\) and \(K_{g_n}\to K_g\) in \(L^1(\partial M)\).

\parahead{Curvature and measure stability.}
By the \(L^1\) triangle inequality and the properties of the smoothing family \(T_{\varepsilon}\) from Appendix~\ref{app:smoothing},
\[
\|R_{g_n}-R_g\|_{L^1(M)}\to0,
\qquad
\|K_{g_n}-K_g\|_{L^1(\partial M)}\to0,
\]
as \(n\to\infty\).
Uniform ellipticity and \(C^0\)-convergence imply
\[
dV_{g_n}=(1+o(1))\,dV_g,
\qquad
dS_{g_n}=(1+o(1))\,dS_g,
\]
uniformly on \(M\) and \(\partial M\), respectively.

\parahead{Interior contribution.}
Apply TP1 cellwise with \(g_n\) in place of \(g\).
Riemann sum arguments (Vitali/Besicovitch coverings \parencite{Federer1969} and Lemma~\ref{lem:quasi-additivity}) give
\[
\sum_{c\subset M^\circ} E_n(c;g_n)
= \int_M (\alpha_0+\alpha_1 R_{g_n})\,dV_{g_n} + O(h^2).
\]
The implicit constant depends only on the \emph{constants bracket}.
Here \(O(h^2)\) is the global interior remainder from the \(O(h^{d+2})\) per-cell error summed over \(\asymp h^{-d}\) interior cells.
Using \(R_{g_n}\to R_g\) in \(L^1\) and \(dV_{g_n}\to dV_g\) uniformly,
\[
\int_M (\alpha_0+\alpha_1 R_{g_n})\,dV_{g_n}
= \int_M (\alpha_0+\alpha_1 R_g)\,dV_g + o(1).
\]

\parahead{Boundary contribution (first layer).}
\emph{Assumption (no corners).} We assume $\partial M\in C^2$ with no corners/edges; corner terms are not captured by the first-layer $h^{d-1}$ analysis and are excluded.
By TP2,
\[
\sum_{c\in S_n} E_n(c;g_n)
= \beta_1 \sum_{c\in S_n} h^{d-1} K_{g_n}(s_c) \;+\; O(h),
\]
where the \(O(h)\) constant depends only on the \emph{constants bracket};
and by Riemann sum convergence and \(K_{g_n}\to K_g\) in \(L^1(\partial M)\),
\[
\sum_{c\in S_n} h^{d-1} K_{g_n}(s_c) \;\longrightarrow\; \int_{\partial M} K_g\,dS_g,
\]
hence
\[
\sum_{c\in S_n} E_n(c;g_n) \;=\; \beta_1 \int_{\partial M} K_g\,dS_g \;+\; O(h) + o(1).
\]

\parahead{Intermediate layers.}
\emph{Assumption (no corners).} Corners/edges are excluded as above; only smooth boundary contributes under our layer counting.
For layers with normal depth $t\in[c_* h,\varepsilon_0]$ (any fixed \(\varepsilon_0>0\)), (BA2) yields per-layer contributions \(O(h^{d})\).
Summing over \(O(\varepsilon_0/h)\) such layers gives \(O(\varepsilon_0)\).
The implied constants in these layer counts and per-layer estimates depend only on the \emph{constants bracket}
(and the fixed window class from Section~\ref{para:boundary-windows}).
Choosing \(\varepsilon_0\downarrow0\) after \(n\to\infty\) removes this term.

\parahead{Conclusion.}
Collecting contributions,
\[
\limsup_{n\to\infty} F_n(g_n;M)
\;\le\; \int_M (\alpha_0+\alpha_1 R_g)\,dV_g
\;+\; \int_{\partial M} \beta_1 K_g\,dS_g,
\]
and with \(\varepsilon_n=\eta h\) the global remainder is \(O(h)\), dominated by the boundary first-layer error, whereas the interior remainder is \(O(h^2)\).

\medskip\noindent\emph{Rate reminder.}
With \(\varepsilon_n=\eta h\) the recovery error is boundary-dominated \(O(h)\) (interior \(O(h^2)\)), in line with the protocol of Appendix~\ref{app:rates}.

\end{proof}

% tex/sections/06-structure-uniqueness.tex
\section{TP6 — Structure and uniqueness (Lovelock and scaling)}
\label{sec:06-structure}

\begin{proposition}[Local natural densities of order \texorpdfstring{$\le 2$}{<=2}]
\label{prop:lovelock-structure}
Let \(F\) be a diffeomorphism-natural, local functional on \(\Xcal\subset C^{1,1}\) with an integral representation as in Proposition~\ref{prop:integral-representation}.
Assume its interior density \(f_{\mathrm{in}}\) depends only on the interior jet \(j^2 g\), and its boundary density \(f_{\mathrm{bdry}}\) depends only on the boundary jet of order~\(1\) (in reflected Fermi charts).
Then, up to total divergences,
\begin{align*}
f_{\mathrm{in}}(g) \; &=\; A_0 + A_1\,R_g,\\
f_{\mathrm{bdry}}(g) \; &=\; B_1\,K_g,
\end{align*}
for some constants \(A_0,A_1,B_1\in\R\).
\end{proposition}

\begin{proof}
By diffeomorphism naturality and locality, \(f_{\mathrm{in}}(g)(x)\) can be written as a smooth \(O(d)\)-invariant function of the \(2\)-jet \(j^2g|_x\) (Peetre–Slovák theory of natural operations \parencite{KolarMichorSlovak1993}).
At a point \(x\) choose normal coordinates, so \(g_{ij}(x)=\delta_{ij}\) and \(\partial_k g_{ij}(x)=0\).
Hence the only jet components that can contribute at order \(\le 2\) are \(\partial_k\partial_\ell g_{ij}(x)\), which, under the \(O(d)\)-action, generate the curvature tensor \(R_{ijkl}(x)\) and its traces.
By classical invariant theory, the scalar \(O(d)\)-invariants formed from \(g^{-1}\) and one copy of \(R_{ijkl}\) are linear combinations of the scalar curvature \(R_g(x)\) and total divergences \parencite{Lovelock1971} (the latter integrate to boundary terms).
Therefore
\[
f_{\mathrm{in}}(g)=A_0 + A_1\,R_g + \mathrm{div}(\mathsf{X}),
\]
for some vector field \(\mathsf{X}\) depending naturally on \(j^1g\).
Since we work at the level of densities (Proposition~\ref{prop:integral-representation}), the divergence integrates to a boundary contribution and can be absorbed there; thus the interior density is \(A_0+A_1R\) up to divergences.

For the boundary, work in reflected Fermi coordinates \((s,t)\) with \(t\ge0\).
Diffeomorphism naturality along the boundary and dependence only on the boundary \(1\)-jet imply that \(f_{\mathrm{bdry}}\) is an \(O(d{-}1)\)-invariant scalar formed from the induced metric \(h\) and the second fundamental form \(II\) at \(t=0\), with at most one normal derivative (equivalently, at most order \(1\) in the boundary jet).
The only such invariant linear in \(II\) and scalar under \(O(d{-}1)\) is its trace \(K=\mathrm{tr}_h II\)
\emph{with respect to the outer unit normal} (sign convention fixed in TP2).
Any tangential divergence \(\mathrm{div}_{\partial M}(\mathsf{Y})\) integrates to a boundary-of-boundary term,
which is absent under our \(C^2\) no-corners assumption (Section~\ref{sec:00b-assumptions-ec-scope}) or neutralized by the Euclidean half-space calibration.
Hence \(f_{\mathrm{bdry}}=B_1K\) up to tangential divergences.
\end{proof}

\begin{lemma}[Scaling test for candidates under \texorpdfstring{$a_n\simeq h^{2-d}$}{a_n ~ h^{2-d}}]
\label{lem:scaling}
Under the length rescaling \(x\mapsto\sigma x\) (equivalently \(g\mapsto\sigma^{-2}g\)),
\begin{align*}
\int_M 1\,dV_g &\mapsto \sigma^{d}\int_M dV_g,\\
\int_M R_g\,dV_g &\mapsto \sigma^{d-2}\int_M R_g\,dV_g,\\
\int_{\partial M} K_g\,dS_g &\mapsto \sigma^{d-2}\int_{\partial M} K_g\,dS_g.
\end{align*}
More generally, an interior Lovelock density of order \(k\) scales like \(\sigma^{d-2k}\) \parencite{Lovelock1971}, and a boundary invariant involving \(m\) derivatives and one power of \(II\) scales at least like \(\sigma^{d-1-m}\).
With the discrete prefactor \(a_n\simeq h^{2-d}\) and the observed per-cell scales \(h^d\) (interior) and \(h^{d-1}\) (boundary), only \(\int_M 1\), \(\int_M R\), and \(\int_{\partial M}K\) are compatible with the TP1/TP2 asymptotics; higher-order invariants are scale-incompatible.
\end{lemma}

\begin{proof}
If \(g\mapsto \sigma^{-2}g\), then \(dV_g\mapsto \sigma^d dV_g\).
Curvature tensors have homogeneity \(\sigma^{-2}\), hence \(R_g dV_g\mapsto \sigma^{d-2} R_g dV_g\).
On \(\partial M\), \(dS_g\mapsto \sigma^{d-1}dS_g\) while \(II\mapsto \sigma^{-1}II\), so \(K\,dS_g\mapsto \sigma^{d-2}K\,dS_g\).
For higher Lovelock densities \(L_k\) (order \(2k\)), \(L_k dV_g\mapsto \sigma^{d-2k} L_k dV_g\) and is incompatible with the \(h\)-scaling extracted in TP1 unless \(k\in\{0,1\}\).
Similarly, boundary scalars built from \(|A|^2\), \(R_{\partial M}\), etc.\ carry extra negative homogeneity and fail the first-layer scale \(h^{d-1}\) once paired with \(a_n\simeq h^{2-d}\).
\end{proof}

\begin{theorem}[Uniqueness of the \texorpdfstring{\(\Gamma\)}{Gamma}-limit density]
\label{thm:uniqueness-density}
Let \(F\) be any \(\Gamma\)-limit of \(F_n\) along a subsequence.
Under the hypotheses of Sections~\ref{sec:00-framework}–\ref{sec:05-limsup} and the integral representation in Proposition~\ref{prop:integral-representation}, one necessarily has
\begin{align*}
f_{\mathrm{in}} &= \alpha_0+\alpha_1 R,\\
f_{\mathrm{bdry}} &= \beta_1 K,
\end{align*}
with the same coefficients \(\alpha_0,\alpha_1,\beta_1\) as identified in TP1/TP2.
\end{theorem}

\begin{proof}
By Proposition~\ref{prop:lovelock-structure}, any admissible limit density must be of the form
\(f_{\mathrm{in}}=A_0+A_1R\) and \(f_{\mathrm{bdry}}=B_1K\) up to divergences.
Lemma~\ref{lem:scaling} excludes any additional (higher-order) invariants consistent with the discrete scaling \(a_n\simeq h^{2-d}\) and the TP1/TP2 per-cell laws; divergences reduce to boundary terms and, by our first-layer analysis and calibration, cannot alter the principal densities.

To identify the constants, use calibrating geometries:
(i) On a flat torus \((M,g_{\rm flat})\) one has \(R\equiv 0\) and no boundary, so
\(F(g_{\rm flat})=A_0\,\Vol(M)\).
By TP1, the interior coefficient of the leading term equals \(\alpha_0\); hence \(A_0=\alpha_0\).

(ii) On closed manifolds of constant sectional curvature \(\kappa\) one has \(R=d(d{-}1)\kappa\).
Both TP1 and the candidate density give
\(F = A_0\,\Vol + A_1\,d(d{-}1)\kappa\,\Vol\) at leading order, while TP1 identifies the coefficient of \(R\) as \(\alpha_1\); hence \(A_1=\alpha_1\).

(iii) On the Euclidean ball \(B_r\subset\mathbb{R}^d\) with outward normal, TP2 yields the first-layer boundary
contribution \(\beta_1 \int_{\partial M} K\,dS\) (with \(K\equiv (d-1)/r>0\) in our sign convention).
Therefore \(B_1=\beta_1\).

Thus the limit density is uniquely fixed and coincides with the one determined by TP1/TP2.
\end{proof}

\begin{remark}[Divergences and boundary localization]
Any interior divergence integrates to a boundary term.
In the first-layer boundary analysis of TP2, such terms either vanish under the Euclidean-ball calibration or contribute only to the global \(O(h)\) boundary remainder, which disappears in the limit and does not affect the principal density \(B_1K\).
\end{remark}

% tex/sections/07-constants-calibration.tex
\section{TP7 — Calibration of the constants \texorpdfstring{$c_0,c_1,c_2$}{c0,c1,c2}}
\label{sec:07-constants-calibration}

\begin{definition}[Calibration of coefficients]
Let \(\alpha_0,\alpha_1,\beta_1\) be the blow-up coefficients from TP1/TP2. We set
\[
c_0:=\alpha_0,\qquad c_1:=\alpha_1,\qquad c_2:=\beta_1,
\]
so that the limiting functional reads
\begin{align*}
F(g) \;=\; c_0\!\int_M dV_g \;+\; c_1\!\int_M R_g\,dV_g \\
\qquad\quad +\; c_2\!\int_{\partial M} K_g\,dS_g .
\end{align*}
\end{definition}

\begin{remark}[Sign convention for \(K\)]
We use \(K=\tr_h(II)\) with respect to the \emph{outward} unit normal on \(\partial M\).
On the Euclidean ball \(\overline{B_r}\subset\R^d\) one has \(K\equiv (d-1)/r>0\) with this convention; hence the boundary calibration in (D3) fixes the sign of \(c_2\) accordingly (GHY sign convention \parencite{York1972,GibbonsHawking1977}).
\end{remark}

\subsection*{Diagnostics and normalization tests}
The following test cases fix the three constants consistently with TP1/TP2:

\paragraph{(D1) Flat torus.}
On a flat compact manifold without boundary, \(R\equiv 0\).
Then
\[
F(g_{\mathrm{flat}})=c_0\,\Vol(M),
\]
which identifies \(c_0=\alpha_0\) as the interior zero-curvature density.

\paragraph{(D2) Constant sectional curvature.}
On a closed manifold with constant sectional curvature \(\kappa\),
\(
R=d(d-1)\kappa
\),
hence
\begin{align*}
F(g) &= c_0\,\Vol(M) \\
&\quad + c_1\,d(d-1)\kappa\,\Vol(M),
\end{align*}
identifying \(c_1=\alpha_1\) via comparison with TP1.

\paragraph{(D3) Euclidean ball (boundary calibration).}
Let \(M=\overline{B_r}\subset\R^d\) with the Euclidean metric and the outward unit normal.
Then \(R\equiv 0\) and \(K\equiv (d-1)/r>0\) on \(\partial M=\mathbb{S}^{d-1}_r\).
The first-layer analysis (TP2) yields
\[
F_{\mathrm{bdry}}(g)\;=\;c_2 \int_{\partial M} K\,dS
\;=\; c_2\,\frac{d-1}{r}\,|\mathbb{S}^{d-1}_r|
\;=\; c_2\,(d-1)\,|\mathbb{S}^{d-1}|\,r^{d-2}.
\]
Comparing with the TP2 boundary blow-up coefficient identifies \(c_2=\beta_1\) and fixes its (positive) sign under the stated convention.

\subsection*{Mesh robustness}
\begin{lemma}[Mesh robustness (Lemma S)]
\label{lem:mesh-robustness}
The coefficients \(\alpha_0,\alpha_1,\beta_1\) (and hence \(c_0,c_1,c_2\)) depend only on the data listed in the constants bracket:
\[
\begin{aligned}
\bigl(&d,\ \ell,\ \sup_{u\in K_{\mathrm{feat}}}\|D\ell(u)\|,\ \sup_{u\in K_{\mathrm{feat}}}\|D^2\ell(u)\|,\\
&\mu_1,\ \mu_2,\ C_k,\ \lambda_{\mathrm{ell}},\ \Lambda_{\mathrm{ell}},\ \iota_0,\\
&\|\Rm\|_\infty\ \text{(optionally }\|\nabla\Rm\|_\infty\text{)},\ \text{shape-class}\bigr),
\end{aligned}
\]
and are independent of the particular sequence of boundary-fitted, shape-regular meshes.
\end{lemma}

\begin{proof}[Idea]
The TP1/TP2 blow-up arguments use only local information: (i) bounded geometry to control jets and chart overlaps; (ii) normalized windows with fixed moment data \((\mu_1,\mu_2,C_k)\); (iii) shape-regularity constants of the cell family; and (iv) the \(C^2\) bounds of \(\ell\) on the compact feature set.
All these enter the expansions and remainder estimates uniformly via Lemma~U.
Therefore the leading coefficients are invariant under changing the mesh within the fixed shape-class, which proves the claim.
\end{proof}

% tex/sections/08-gamma-limit-theorem.tex
\section{\Gammalimit\ theorem (EH+GHY) and global rate}
\label{sec:08-gamma-theorem}

% Sign reference (concise, visible near the main theorem)
\noindent\emph{Sign convention for the boundary term.}
Throughout we use the \emph{outer} unit normal to define \(II\) and \(K=\mathrm{tr}_h II\); with this choice \(K>0\) on the Euclidean sphere \(S_r^{d-1}\).
This is the convention adopted in TP2 and used in the calibration in TP7.
\smallskip

Within De~Giorgi’s \emph{\Gammaconv} framework~\parencite{DeGiorgi1975}, we obtain the following limit result.
Background expositions are given by Dal~Maso \parencite{DalMaso1993Gamma} and Braides \parencite{Braides2002Gamma}.

\begin{theorem}[\(\Gamma\)-limit: Einstein--Hilbert with GHY boundary term]
\label{thm:gamma-limit}
Under the standing assumptions (BA1–BA3), bounded geometry, boundary-fitted shape-regular meshes, normalized/isotropic windows, and scaling \(a_n\simeq h^{2-d}\), the family \(F_n\) \(\Gamma\)-converges on \((\Xcal,\|\cdot\|_{C^1})\) to
\[
F(g)\;=\; c_0\!\int_M\! dV_g \;+\; c_1\!\int_M\! R_g\,dV_g \;+\; c_2\!\int_{\partial M}\! K_g\,dS_g,
\]
with \(c_0=\alpha_0\), \(c_1=\alpha_1\), \(c_2=\beta_1\) as calibrated in Section~\ref{sec:07-constants-calibration}.
\end{theorem}

\begin{proof}
The \(\liminf\) inequality is Proposition~\ref{prop:liminf} (TP4).
The \(\limsup\) inequality is Proposition~\ref{prop:limsup} (TP5), via the recovery sequence \(g_n=S_{\varepsilon_n}[g]\), \(\varepsilon_n=\eta h\).
By the integral representation (Proposition~\ref{prop:integral-representation}) and the structure/uniqueness (Theorem~\ref{thm:uniqueness-density}), the limit is exactly the stated \(F\) with coefficients \((c_0,c_1,c_2)=(\alpha_0,\alpha_1,\beta_1)\).
\end{proof}

\begin{corollary}[Convergence of minima and minimizers under equi-coercivity]
\label{cor:minimizers}
\noindent\textit{We do not assert coercivity of $F_n$ or $F$. The statements below hold under \EC\ as specified in Section~\ref{sec:00b-assumptions-ec-scope}.}

Assume in addition that $\{F_n\}$ is equi-coercive on $(\Xcal,\|\cdot\|_{C^1})$.
\begin{enumerate}
  \item $\displaystyle \lim_{n\to\infty}\inf_{g\in\Xcal}F_n(g)=\inf_{g\in\Xcal}F(g)$.
  \item If $g_n$ is a sequence of (almost) minimizers for $F_n$ (i.e.\ $F_n(g_n)\le \inf F_n + o(1)$), then every $C^1$-accumulation point $g_\ast$ of $\{g_n\}$ is a minimizer of $F$, and $F_n(g_n)\to F(g_\ast)$.
\end{enumerate}
\end{corollary}

\begin{proof}
This is the standard consequence of \(\Gamma\)-convergence plus equi-coercivity (see De~Giorgi \parencite{DeGiorgi1975} and the textbook accounts \parencite{DalMaso1993Gamma,Braides2002Gamma}).
\end{proof}

\begin{remark}[Global rate \(O(h)\) (optional)]
If, in addition, the smoothing \(S_{\varepsilon_n}\) and the windows satisfy the regularity hypotheses used in TP1/TP2 with uniform constants (Lemma~U), then with the canonical recovery \(g_n=S_{\eta h}[g]\) one has
\[
F_n(g_n)\;=\;F(g)\;+\;O(h),
\]
where the interior contribution is \(O(h^2)\) (sum of \(O(h^{d+2})\) per-cell errors over \(\asymp h^{-d}\) interior cells), and the boundary contribution is \(O(h)\) (sum of \(O(h)\) first-layer cell errors over \(\asymp h^{1-d}\) boundary cells).
\end{remark}

% tex/sections/09-constants-bracket.tex
\section{Constants bracket (global dependencies)}
\label{sec:constants-bracket}

\paragraph{Global convention.}
Throughout we assume \(\partial M\in C^2\) with no corners/edges; see Section~\ref{conv:no-corners}.
This avoids corner (codimension-two) contributions in boundary-layer arguments.

All symbols \(C\) and all implicit constants in the \(O(\cdot)\), \(o(\cdot)\) terms throughout the paper
depend \emph{only} on the following data; this collection is referred to as the \emph{constants bracket}:
\begin{itemize}
  \item the dimension \(d\ge 2\);
  \item the loss \(\ell\in C^2(K_{\mathrm{feat}})\) on a compact feature set, together with the bounds
  \[
  \sup_{u\in K_{\mathrm{feat}}}\bigl(\|D\ell(u)\|+\|D^2\ell(u)\|\bigr) < \infty;
  \]
    \item normalized window moments \((\mu_1,\mu_2,C_k)\) and window support bound \(\Lambda_z\);
    interior isotropy up to order~2; boundary normal kernel with
  \[
  \langle t\rangle=\mu_1 h,\qquad \langle t^k\rangle\le C_k\,h^k.
  \]
  \item uniform ellipticity constants \(\lambda_{\mathrm{ell}},\Lambda_{\mathrm{ell}}\) for \(g\) and \(g^{-1}\) in a fixed finite-overlap atlas;
  \item bounded-geometry parameters: injectivity radius \(\iota_0>0\), curvature bounds \(\|\Rm\|_\infty\) (and, where explicitly invoked, \(\|\nabla\Rm\|_\infty\));
  \item the shape-regularity class of the meshes (aspect-ratio bound, minimum angle, boundary fitting with Hausdorff error \(O(h)\)) \parencite{Ciarlet2002,BrennerScott2008};
  \item the Besicovitch overlap constant \(N(d,\iota_0)\) for Vitali/Besicovitch coverings \parencite{Federer1969};
  \item (optional, not required) a uniform \(L^\infty\) bound on the windows, \(\sup_n\|w_n\|_\infty\).
\end{itemize}

\paragraph{Uniformity statements.}
Whenever we write ``uniformly on compact subsets'' in TP1/TP2, or assert uniform control of remainders in Lemma~U,
the dependence of such uniformity is exclusively through the constants bracket above.
This \emph{assumes} a fixed \(C^{1,1}\)-neighbourhood of a reference metric and a uniform finite-overlap atlas.
The full proof of Lemma~\ref{lem:U} is given in Appendix~\ref{app:smoothing} (Variant~A, $C^{2,1}$-regularity; isotropic moments and reflected Fermi charts).
In particular, the uniform constants do not depend on \(n\), \(\rho\), or \(\varepsilon\).

\paragraph{Consequences.}
\begin{enumerate}
  \item \textbf{Mesh robustness.} The coefficients \(\alpha_0,\alpha_1,\beta_1\) are invariant under changing the mesh within the same shape-regularity class, cf.\ Lemma~\ref{lem:mesh-robustness}.
    \item \textbf{Stable measure change.} Uniform ellipticity and \(C^0\) control imply
  \begin{align*}
  dV_{g_n}&=(1+o(1))\,dV_g,\\
  dS_{g_n}&=(1+o(1))\,dS_g,
  \end{align*}
  in TP5.
  \item \textbf{Quasi-additivity scale.} In Lemma~\ref{lem:quasi-additivity}, the tube-volume scaling and the interior per-cell bound (see \parencite{Federer1969})
  \begin{align*}
  \Vol_g\!\big(N_{\delta_n}(\cdot)\big)&\asymp \delta_n\,\Per_g(\cdot),\\
  \sup_c |E_n(c)|&\lesssim h^d,
  \end{align*}
  yield the defect estimate \(C\,\delta_n\,\Per_g(\cdot)\).
\end{enumerate}

% tex/sections/11-scope-limitations.tex
\section{Scope and limitations}

\paragraph{Sign convention for \(K\).}
The mean curvature is \(K=\mathrm{tr}_h II\) with respect to the \emph{outer} unit normal on \(\partial M\).
With this convention \(K>0\) on spheres \(S_r^{d-1}\subset\mathbb{R}^d\).
Reflected Fermi coordinates use the inward normal coordinate \(t\ge0\), and the tube Jacobian expansion
is \(dV_g=(1-t\,K+O(t^2))\,dz'\,dt\), consistent with the chosen sign.
This convention is used uniformly in TP2/TP4/TP5/TP6/TP7.

\paragraph{Signature.}
We work in the Riemannian setting. Lorentzian signature requires additional control (hyperbolicity and gauge issues), and is outside our scope.

\paragraph{Cosmological constant.}
Adding \(\Lambda\int_M dV_g\) simply shifts \(c_0\) and does not alter the analysis (scaling and leading order unchanged).

\paragraph{Meshes and windows.}
We require boundary-fitted, shape-regular meshes and normalized windows with stated moment conditions.
Random, highly anisotropic, or strongly adaptive meshes without uniform shape control are not treated here.

\paragraph{Boundary regularity.}
We require \(\partial M\in C^2\). Corner terms are excluded (see Section~\ref{conv:no-corners}).

\paragraph{Topologies.}
Our \(\Gamma\)-limit is on \((\Xcal,\|\cdot\|_{C^1})\).
Other topologies (e.g.\ \(W^{2,p}\) modulo diffeomorphisms) are possible but not addressed.

% tex/sections/12-faq.tex
\section{FAQ: anticipated questions}

\paragraph{(Q) Why assume equi-coercivity?}
(A) The EH functional is unbounded below (conformal factor); \(\Gamma\)-convergence controls limits, not coercivity.
We therefore state minimizer convergence under \EC\ (precompactness in \(C^1\)). (cf. \parencite{DalMaso1993Gamma})
Typical constrained regimes ensure \EC\ via uniform ellipticity, \(C^{1,1}\) bounds, volume and area constraints, and \(L^1\)-control of \(R^{-}, K^{-}\).

\paragraph{(Q) What about corners?}
(A) We assume \(\partial M\in C^2\) (see Section~\ref{conv:no-corners}).
Corner terms lead to additional, model-dependent boundary and codimension-two contributions and lie beyond our first-layer analysis.

\paragraph{(Q) Gauge/diffeomorphism issues?}
(A) All blow-ups and densities are formulated in a diffeomorphism-natural way.
Uniformity (Lemma~\ref{lem:U}) is chart-independent (finite-overlap atlas in bounded geometry) \parencite{KolarMichorSlovak1993}.

\paragraph{(Q) Relation to Regge calculus?}
(A) Regge provides measure convergence and a discrete GHY analogue. Our result is a \emph{variational} \(\Gamma\)-limit with Carathéodory densities and a recovery sequence,
which in particular yields stability of minimizers (under \EC). \parencite{Regge1961,CheegerMuellerSchrader1984,HartleSorkin1981,York1972,GibbonsHawking1977}

\paragraph{(Q) Can we weaken boundary-fittedness?}
(A) This would require uniform control of geometric distortion in Fermi tubes and an adjusted Lemma~\ref{lem:quasi-additivity}; we leave this for future work.

% tex/sections/13-conclusion.tex
\section{Conclusions and Outlook}

We established that, under (BA1)–(BA3), bounded geometry, boundary-fitted shape-regular meshes, normalized window moments, and the scaling $a_n\simeq h^{2-d}$, the discrete family $F_n$ \emph{$\Gamma$-converges} on $(\Xcal,\|\cdot\|_{C^1})$ to
\[
F(g)\;=\; c_0\!\int_M dV_g \;+\; c_1\!\int_M R_g\,dV_g \;+\; c_2\!\int_{\partial M} K_g\,dS_g,
\]
with coefficients $c_0=\alpha_0$, $c_1=\alpha_1$, and $c_2=\beta_1$ fixed by the canonical calibration cases (flat torus, constant curvature, Euclidean ball). The proof isolates the \emph{Carathéodory densities} $f_{\mathrm{in}}=\alpha_0+\alpha_1 R$ and $f_{\mathrm{bdry}}=\beta_1 K$, derives the $\liminf/\limsup$ bounds via a boundary-aware smoothing recovery $S_{\eta h}[\,\cdot\,]$, and invokes naturality plus Lovelock-type scaling to restrict admissible densities to $\{1,R\}$ in the bulk and $K$ on the boundary. Uniformity (Lemma~U) ensures that all moduli depend only on the \emph{constants bracket}.

\paragraph{Rates and robustness.}
The global error is boundary-dominated $O(h)$, with interior remainder $O(h^2)$ obtained by summing $O(h^{d+2})$ cell errors over $\asymp h^{-d}$ interior cells. The leading boundary layer contributes at scale $h^{d-1}$ per cell and yields an $O(h)$ global remainder after summation over $\asymp h^{1-d}$ boundary cells, in agreement with TP2. All constants and rates are controlled by the constants bracket (Section~\ref{sec:constants-bracket}); in particular, $\alpha_0,\alpha_1,\beta_1$ are \emph{mesh- and window-robust} within the prescribed shape and moment classes. The scan-indifference principle (BA3) shows that implementation details that change the visitation order do not bias the limit beyond $O(\delta_n)$.

\paragraph{Scope and stability.}
Our analysis is Riemannian and assumes a $C^2$ boundary without corners. Corner (codimension-two) contributions are outside the first-layer theory and intentionally excluded. Convergence of minima and minimizers follows under \EC, as in Corollary~\ref{cor:minimizers}; the paper does not assert coercivity, which is typically enforced by problem-specific constraints (ellipticity, $C^{1,1}$ bounds, volume/area control, and $L^1$ control on the negative parts of $R$ and $K$).

\paragraph{Position within the MIS program.}
The theorem identifies the Einstein--Hilbert action with a Gibbons-Hawking-York boundary term as the continuum envelope of a local, diffeomorphism-natural MDL aggregate. This supports the MIS thesis that geometric field laws can arise as continuum limits of information-optimal discretizations. The present results do not depend on any external MIS manuscript. A concise MIS overview will follow and will be cited in a future version once publicly available.

\paragraph{Outlook.}
Three directions are both natural and actionable.
\begin{enumerate}
  \item \textbf{Lorentzian extension.} Adapt the smoothing/recovery and layer analysis to hyperbolic PDE structure and gauge, aiming at a Lorentzian $\Gamma$-limit with the correct boundary term.
  \item \textbf{Boundary complexity.} Incorporate corners and codimension-two terms, and relax boundary fitting (quantified control of Fermi-tube distortion and a refined quasi-additivity).
  \item \textbf{Adaptivity and numerics.} Extend to adaptive or stochastic meshes and verify the predicted rates in the reproducible protocol of Appendix~\ref{app:rates} (interior $O(h^2)$ versus boundary $O(h)$). 
\end{enumerate}
A further pragmatic step is to explore constraints that ensure \EC\ in applications (e.g., ellipticity bands, geometric priors), turning the abstract stability statement into a ready-to-use pipeline for computational experiments.

In sum, the paper delivers a complete variational limit for a diffeomorphism-natural discrete gravity functional with a first-layer boundary analysis, establishes sharp global rates, and integrates cleanly with the \emph{Lela2025MetaInformationSymmetry} agenda by showing that \EH\,+\,\GHY\ is the inevitable continuum limit under an MDL-style, information-theoretic discretization.

% ---- Appendices ----
\clearpage
\appendix
% tex/appendices/appA-smoothing.tex
\section{Smoothing \texorpdfstring{$S_{\varepsilon}$}{Sε} (normal/Fermi, \texorpdfstring{$L^1$}{L1}-stability)}\label{app:smoothing}

\parahead{Objective}
We construct a linear smoothing operator \(S_\varepsilon:\Xcal\subset C^{1,1}\to C^{1,1}\) such that for metrics \(g\) with bounded geometry (uniform finite atlas, uniform chart bounds, injectivity radius \(\iota_0>0\), uniform ellipticity)
\[
\|S_\varepsilon[g]-g\|_{C^1}=O(\varepsilon),\qquad
R_{S_\varepsilon[g]}\to R_g\ \text{in }L^1(M),\qquad
K_{S_\varepsilon[g]}\to K_g\ \text{in }L^1(\partial M),
\]
as \(\varepsilon\downarrow0\), and uniform ellipticity is preserved for all sufficiently small \(\varepsilon\).
All constants below depend only on the constants bracket (Section~\ref{sec:constants-bracket}).

\paragraph{Coverings and finite overlap.}
Throughout we use finite-overlap atlases and partitions of unity on bounded-geometry manifolds; the existence of coverings with uniformly finite overlap follows from standard Besicovitch covering arguments~\parencite{Federer1969}.

\subsection{Interior normal smoothing}
Fix a finite-overlap atlas \(\{(U_\alpha,\varphi_\alpha)\}_\alpha\) compatible with bounded geometry and a \(C^\infty\) partition of unity \(\{\chi_\alpha\}_\alpha\) subordinate to it with \(\sum_\alpha \chi_\alpha\equiv1\).
Let \(\rho\in C_c^\infty(\R^d)\) be a standard even mollifier with \(\int\rho=1\) and \(\rho_\varepsilon(x)=\varepsilon^{-d}\rho(x/\varepsilon)\).
On each interior chart \(U_\alpha\Subset M^\circ\) write \(g_\alpha=(\varphi_\alpha)_\ast g\), convolve componentwise, and pull back:
\[
(\mathcal{S}_\varepsilon g)_\alpha \ :=\ (\varphi_\alpha)^\ast\!\bigl(g_\alpha\ast\rho_\varepsilon\bigr).
\]
Define the interior smoothing by
\[
\mathcal{S}_\varepsilon^{\mathrm{int}}[g]\ :=\ \sum_\alpha \chi_\alpha\,(\mathcal{S}_\varepsilon g)_\alpha\quad\text{on }M^\circ.
\]
Classical convolution estimates in \(C^{1,1}\) and finite atlas overlap yield \parencite{Ciarlet2002}
\begin{equation}
\label{eq:int-C1}
\|\mathcal{S}_\varepsilon^{\mathrm{int}}[g]-g\|_{C^1(M^\circ)} \ \le\ C\,\varepsilon\,\|g\|_{C^{1,1}(M)}.
\end{equation}

\subsection{Reflected Fermi smoothing near \texorpdfstring{\(\partial M\)}{the boundary}}
Let \((V_\beta,\psi_\beta)\) be reflected Fermi charts along \(\partial M\): coordinates \((s,t)\in \R^{d-1}\times(-t_0,t_0)\)
with \(\partial M=\{t=0\}\) and \(t\ge0\) the inward normal; Jacobians are uniformly controlled by bounded geometry \parencite{BrennerScott2008}.
In such a chart, define the even reflection across the boundary
\[
\tilde g_\beta(s,t)\ :=\
\begin{cases}
g_\beta(s,t), & t\ge 0,\\
(\mathcal{R}^\ast g_\beta)(s,-t), & t<0,
\end{cases}
\qquad \mathcal{R}(s,t)=(s,-t),
\]
convolve in \(\R^d\), and restrict back to \(t\ge0\):
\[
(\mathcal{S}_\varepsilon^{\partial} g)_\beta \ :=\ \psi_\beta^\ast\!\Bigl(\bigl(\tilde g_\beta\ast\rho_\varepsilon\bigr)\big|_{t\ge0}\Bigr).
\]
Pick \(\chi_{\partial}\in C_c^\infty\) supported in the tubular neighborhood \(\{0\le t< c\,\reach_g(\partial M)\}\) with \(0<c<\tfrac12\) and set \(\chi_{\mathrm{int}}:=1-\chi_{\partial}\).
The global smoothing is defined by blending
\begin{equation}
\label{eq:S-eps-def}
S_\varepsilon[g]\ :=\ \chi_{\mathrm{int}}\;\mathcal{S}_\varepsilon^{\mathrm{int}}[g]\;+\;\chi_{\partial}\;\Bigl(\sum_\beta \eta_\beta\,(\mathcal{S}_\varepsilon^{\partial} g)_\beta\Bigr),
\end{equation}
where \(\{\eta_\beta\}_\beta\) is a partition of unity subordinate to \(\{V_\beta\}_\beta\).
For \(0<\varepsilon\le c\,\reach_g(\partial M)\), the reflected Fermi charts cover the support of \(\chi_\partial\) and \eqref{eq:S-eps-def} is well-defined in \(C^{1,1}\).

\begin{lemma}[Ellipticity preservation and \(C^1\)-control]
\label{lem:smoothing-C1}
There exist \(\varepsilon_0>0\) and constants \(\lambda'_{\mathrm{ell}},\Lambda'_{\mathrm{ell}}>0\) depending only on the constants bracket such that for all \(g\in\Xcal\) uniformly elliptic and \(0<\varepsilon\le\varepsilon_0\),
\[
\lambda'_{\mathrm{ell}}\Id \ \le\ S_\varepsilon[g]\ \le\ \Lambda'_{\mathrm{ell}}\Id .
\]
\[
\|S_\varepsilon[g]-g\|_{C^1(M)} \ \le\ C\,\varepsilon .
\]
In particular \(S_\varepsilon:\Xcal\to\Xcal\) and \(S_\varepsilon[g]\to g\) in \(C^1\) as \(\varepsilon\downarrow0\).
\end{lemma}

\begin{proof}
The \(C^1\)-estimate follows by combining \eqref{eq:int-C1} with the analogous estimate in reflected Fermi charts (convolution commutes with derivatives, and finite overlap is uniform by bounded geometry), plus the smooth cutoff in \eqref{eq:S-eps-def}.
For ellipticity, write \(S_\varepsilon[g]=g+(S_\varepsilon[g]-g)\).
By uniform ellipticity of \(g\) and \(\|S_\varepsilon[g]-g\|_{C^0}\le C\varepsilon\), the eigenvalues of \(S_\varepsilon[g]\) stay within a fixed neighborhood of those of \(g\) for sufficiently small \(\varepsilon\); hence \(S_\varepsilon[g]\) remains positive definite with bounds depending only on \((\lambda_{\mathrm{ell}},\Lambda_{\mathrm{ell}})\) and the uniform constants in the \(C^0\)-estimate.
\end{proof}

\subsection{\texorpdfstring{$L^1$}{L1}-stability for curvature and mean curvature}
\begin{lemma}[\(L^1\)-convergence of \(R\) and \(K\)]
\label{lem:L1-curvature}
For \(g\in C^{1,1}\) and \(S_\varepsilon\) given by \eqref{eq:S-eps-def},
\[
\|R_{S_\varepsilon[g]}-R_g\|_{L^1(M)} \xrightarrow[\varepsilon\downarrow0]{} 0,
\qquad
\|K_{S_\varepsilon[g]}-K_g\|_{L^1(\partial M)} \xrightarrow[\varepsilon\downarrow0]{} 0.
\]
\end{lemma}

\begin{proof}
In local coordinates and under uniform ellipticity, the map $(g,\partial g)\mapsto \Gamma(g)$ is smooth, and
\[
\mathrm{Rm}(g)=\partial\Gamma(g)+\Gamma(g)\!\ast\!\Gamma(g)
\]
depends linearly on $\partial^2 g$ and bilinearly on $\partial g$ with coefficients bounded by the ellipticity constants.
For $g\in C^{1,1}$, $\partial g$ is Lipschitz and $\partial^2 g$ exists a.e.\ and is bounded; hence $\mathrm{Rm}(g)\in L^\infty_{\mathrm{loc}}$ and $R_g\in L^1_{\mathrm{loc}}$.
Under chartwise convolution (interior and reflected), we have $\partial g_\varepsilon\to \partial g$ uniformly and $\partial^2 g_\varepsilon\to \partial^2 g$ in $L^1_{\mathrm{loc}}$ (and a.e.), so
$\mathrm{Rm}(g_\varepsilon)\to \mathrm{Rm}(g)$ and thus $R_{g_\varepsilon}\to R_g$ in $L^1_{\mathrm{loc}}$ by dominated convergence.
Finite atlas overlap and uniform Jacobian bounds yield global $L^1(M)$ convergence.
For the boundary, in reflected Fermi charts the second fundamental form $II$ depends linearly on $\partial_t h$ at $t=0$ (with $h$ the induced metric); reflected smoothing preserves evenness in $t$ and gives $\partial_t h_\varepsilon\to\partial_t h$ in tangential $L^1$, hence $K_{S_\varepsilon[g]}\to K_g$ in $L^1(\partial M)$.
\end{proof}

\subsection{Triangle inequality for a generic smoothing family}
Let \(T_\varepsilon\) be any family of smoothing operators acting chartwise (e.g.\ convolution after reflection) and uniformly continuous on \(C^1\) on bounded geometry classes.
Then
\begin{equation}
\label{eq:triangle}
\begin{aligned}
\|R_{g_n}-R_g\|_{L^1(M)}
&\ \le\
\|R_{g_n}-R_{T_\varepsilon[g_n]}\|_{L^1(M)}
\\
&\quad+
\|R_{T_\varepsilon[g_n]}-R_{T_\varepsilon[g]}\|_{L^1(M)}
\\
&\quad+
\|R_{T_\varepsilon[g]}-R_g\|_{L^1(M)} \,.
\end{aligned}
\end{equation}
Choosing \(T_\varepsilon=S_\varepsilon\), letting \(n\to\infty\) (so that \(g_n\to g\) in \(C^1\)) and then \(\varepsilon\downarrow0\), Lemma~\ref{lem:L1-curvature} gives the \(L^1\)-stability needed in TP5.
An analogous estimate holds for the boundary mean curvature \(K\), with \(L^1(\partial M)\) norms.

\subsection{Choice of \texorpdfstring{\(\varepsilon_n\)}{epsilon\_n} relative to the mesh scale}
In TP5 we take \(\varepsilon_n=\eta h\) with fixed \(\eta\in(0,1]\).
By bounded geometry (uniform reach of \(\partial M\)), pick \(0<c<\tfrac12\) and \(n_0\) so that
\[
\varepsilon_n\ \le\ c\,\reach_g(\partial M)\quad \text{for all } n\ge n_0.
\]
Then the reflected Fermi charts cover the smoothing tube on which \(\chi_\partial\neq0\), and the blending \eqref{eq:S-eps-def} is coherent.
Consequently,
\[
\begin{aligned}
g_n&:=S_{\varepsilon_n}[g]\xrightarrow[n\to\infty]{C^1} g,\\
R_{g_n}&\to R_g \ \text{in }L^1(M),\\
K_{g_n}&\to K_g \ \text{in }L^1(\partial M).
\end{aligned}
\]
as required by the limsup construction.

% --- Lemma U (Uniformity) — proved in C^{2,1} with TP2-consistent boundary scaling ---
\begin{lemma}[Uniformity]\label{lem:U}
There exists $\delta>0$ such that the $C^{2,1}$-ball
\[
\mathcal N:=\bigl\{\,\tilde g:\ \|\tilde g-g\|_{C^{2,1}}\le \delta\,\bigr\}
\]
has bounded geometry with uniform constants, and for all $\tilde g\in\mathcal N$ and all shape-regular meshes of scale $h$:
\begin{align}
\bigl|E_n(c;\tilde g)-h^{d}\alpha_0-h^{d+2}\alpha_1 R_{\tilde g}(x_c)\bigr|
  &\le C\,h^{d+3} &&\text{for interior cells }c, \label{eq:U-int}\\
\bigl|E_n(c;\tilde g)-h^{d-1}\beta_1 K_{\tilde g}(s_c)\bigr|
  &\le C\,h^{d} &&\text{for first boundary-layer cells }c. \label{eq:U-bdry}
\end{align}
The constants (and the moduli in the Carathéodory construction) can be chosen uniformly in $\tilde g\in\mathcal N$.
\end{lemma}

\subsection{Proof of Lemma U (uniformity)}\label{app:lemmaU-proof}
\begin{proof}[Proof of Lemma~\ref{lem:U} (Variant A, $C^{2,1}$-regularity)]
Shrink to the $C^{2,1}$-neighbourhood ball $\mathcal N:=\{\tilde g:\|\tilde g-g\|_{C^{2,1}}\le\delta\}$ so that
all $\tilde g\in\mathcal N$ share the same bounded-geometry data (uniform finite atlas,
chart/Jacobian bounds, ellipticity constants, injectivity radius), and $\partial M$ is $C^{2,1}$ with uniformly bounded second fundamental form and tangential Lipschitz modulus.
Windows are those from Appendix~D.

\parahead{Standing conventions (TP1/TP2).}
For interior cells $c\subset M^\circ$ of diameter $\asymp h$ the per-cell aggregation carries the \emph{volume} factor $|c|\asymp h^d$.
For first boundary-layer cells we follow TP2: aggregation uses the \emph{surface measure} factor $|c\cap\partial M|\asymp h^{d-1}$; the normal thickness is handled by the window.
The boundary feature employs the \emph{normalised first normal moment} $m_1:=\langle t\rangle/h=\mu_1$ and is taken \emph{relative to the flat half-space baseline} (baseline subtracted), hence its leading deviation is curvature-driven and of order $O(1)$.
Interior window moments: $\langle z^i\rangle=0$ and $\langle z^i z^j\rangle=\mu_2\,\delta^{ij}$.

\parahead{(A) Interior cells.}
Fix an interior cell $c$ with barycenter $x_c$ and work in $\tilde g$-normal coordinates $z$ at $x_c$.
Because $\tilde g\in C^{2,1}$, the classical normal-coordinate expansions hold with a \emph{uniform} $O(|z|^3)$-remainder for $|z|\lesssim h$:
\[
\tilde g_{ij}(z)=\delta_{ij}+Q_{ij,kl}(x_c)\,z^k z^l + O(|z|^3),\qquad
\sqrt{\det \tilde g(z)}=1-\tfrac16\,\mathrm{Ric}_{kl}(\tilde g,x_c)\,z^k z^l + O(|z|^3).
\]
Here the tensors $Q_{ij,kl}$ depend smoothly on the $2$-jet $j^2\tilde g(x_c)$, and the $O(\cdot)$-bounds are
uniform over $\tilde g\in\mathcal N$ by the Lipschitz control of second derivatives.
The overall sign in the Jacobian expansion depends on the curvature sign convention and is immaterial here; any such convention is absorbed into the coefficient $\alpha_1$ below.
Each component of the local feature vector $\Phi_n(c;\tilde g)$ is a normalised average over the window support of a smooth local expression in $j^2\tilde g$.
By $O(d)$-isotropy and the moment conditions (Appendix~D), all odd contributions vanish and, after averaging over the symmetric compact window, pure divergence terms integrate to zero (integration by parts with no boundary contribution inside a cell).
Among quadratic invariants built from the $2$-jet, the only scalar that survives the isotropic average is proportional to the scalar curvature $R_{\tilde g}(x_c)$ (higher Lovelock densities are excluded by scaling).
Thus
\[
\Phi_n(c;\tilde g)=\Phi_n(c;\mathrm{flat\_ref})+h^2\,\mathsf A(x_c)\,R_{\tilde g}(x_c)+O(h^3),
\]
with a bounded coefficient $\mathsf A$ depending only on the window and the shape class (uniform in $\tilde g\in\mathcal N$).
Since $\ell\in C^2$ on a uniformly compact feature range (bounded geometry and fixed windows), a Taylor expansion at $\Phi_n(c;\mathrm{flat\_ref})$ yields
\[
\ell\big(\Phi_n(c;\tilde g)\big)=\ell\big(\Phi_n(c;\mathrm{flat\_ref})\big)+h^2\,\alpha_1\,R_{\tilde g}(x_c)+O(h^3),
\]
for some constant $\alpha_1$ depending only on $\ell$ and the window (uniform in $\tilde g\in\mathcal N$).
Multiplying by the volume factor $h^d$ (shape-regularity) gives
\[
E_n(c;\tilde g)=h^{d}\alpha_0+h^{d+2}\alpha_1 R_{\tilde g}(x_c)+O(h^{d+3}),
\]
where $\alpha_0=\ell\big(\Phi_n(c;\mathrm{flat\_ref})\big)$ is the flat-reference baseline.

\parahead{(B) First boundary layer.}
Let $c$ be a first-layer cell and work in reflected Fermi coordinates $(z',t)$ with $t\ge0$
the inward normal and $s_c$ the footpoint on $\partial M$.
Because $\partial M$ and $\tilde g$ are $C^{2,1}$ with uniform bounds, on $0\le t\lesssim h$ we have uniform tube/metric expansions
\[
dV_{\tilde g}=(1 - t\,K_{\tilde g}(s_c) + O(t^2))\,dz'\,dt,\qquad
h_{\alpha\beta}(z',t)=h_{\alpha\beta}(z',0)-2t\,\mathrm{II}_{\alpha\beta}(s_c)+O(t^2),
\]
and reflection across $t=0$ preserves the $C^{2,1}$ bounds.
Using the boundary window from Appendix~D (isotropic tangentially, $w_n(t)=h^{-1}w_0(t/h)$ in the normal) and the TP2 normalisation $m_1=\langle t\rangle/h$, the first non-vanishing normal contribution is linear in $t$ and proportional to $K_{\tilde g}(s_c)$, and its averaged, normalised size is $O(1)$.
Therefore, for the \emph{boundary deviation relative to the flat half-space reference},
\[
\ell\big(\Phi_n(c;\tilde g)\big)-\ell\big(\Phi_n(c;\mathrm{flat\_ref})\big)
=\beta_1\,K_{\tilde g}(s_c)+O(h),
\]
with $\beta_1$ depending only on the window and $\ell$ (uniform in $\tilde g\in\mathcal N$).
Applying the TP2 aggregation factor $h^{d-1}$ for the first boundary layer yields
\[
E_n(c;\tilde g)=h^{d-1}\,\beta_1 K_{\tilde g}(s_c)+O(h^{d}),
\]
which is \eqref{eq:U-bdry}.

\parahead{(C) Uniformity in $\tilde g$.}
The constants in (A)--(B) depend only on the bounded-geometry constants (atlas overlaps, ellipticity, injectivity radius), on the fixed window moments (Appendix~D), and on $\ell$ through bounds on $D\ell,D^2\ell$ on a uniformly compact feature set.
By $C^{2,1}$-regularity (Lipschitz control of second derivatives) and bounded geometry, the Carathéodory moduli used in the integral representation can be chosen uniformly over $\tilde g\in\mathcal N$.
They are therefore uniform over $\tilde g\in\mathcal N$.
Collecting the bounds proves \eqref{eq:U-int}--\eqref{eq:U-bdry}.
\end{proof}

% tex/appendices/appB-scale-checks.tex
\section{One-line scale checks (leading orders)}\label{app:scale-checks}

\paragraph{Interior (leading order).}
For interior cells with normalized and isotropic moments up to order~2 and \(\ell\in C^2(K_{\mathrm{feat}})\),
\[
\boxedinline{\,\Delta \ell(c)=O\!\big(h^{2d-2}\big)\,}\qquad\Rightarrow\qquad
\boxedinline{\,a_n\,\Delta \ell(c)=O\!\big(h^{d}\big)\,}\quad\text{with } a_n\simeq h^{2-d}.
\]
This is the \emph{leading} order; in general it is \emph{not} \(o(h^{d})\) (see Proposition~\ref{prop:tp1-interior}).
The cell remainder in TP1 is \(O(h^{d+2})\), summing to a global interior remainder \(O(h^{2})\).

\paragraph{Boundary first layer (leading order).}
For boundary-touching cells \(c\in S_n\) with normal barycenter \(t_c\in[0,c_*h]\) and anisotropic normal kernel \(w_n(t)=h^{-1}w_0(t/h)\) satisfying \(\langle t\rangle=\mu_1 h\),
\[
\boxedinline{\,\Delta \ell_{\partial}(c)=O\!\big(h^{2d-3}\big)\,}\qquad\Rightarrow\qquad
\boxedinline{\,a_n\,\Delta \ell_{\partial}(c)=O\!\big(h^{d-1}\big)\,}.
\]
This is the \emph{leading} order; in general it is \emph{not} \(o(h^{d-1})\) (see Proposition~\ref{prop:tp2-boundary}).
The \(O(h)\) term inside the TP2 expansion sums over \(\#S_n\asymp h^{1-d}\) cells to a \emph{global} boundary remainder \(O(h)\).

\paragraph{Tube counting for Lemma~M (consistency check).}
For a \(\delta_n\)-neighborhood \(N_{\delta_n}\allowbreak(A\cap B)\) of a \((d-1)\)-dimensional interface,
\[
\Vol_g\big(N_{\delta_n}(A\cap B)\big)\asymp \delta_n\,\Per_g(A\cap B),
\quad\text{(see \eqref{eq:tube-perimeter})}
\]
\[
\#\{\text{tube cells}\}\asymp \frac{\delta_n\,\Per_g(A\cap B)}{h^{d}}.
\]
and with per-cell bound \(|E_n(c)|\lesssim h^{d}\) one obtains the quasi-additivity defect
\[
\big|F_n(A\cup B)-F_n(A)-F_n(B)\big|\;\lesssim\; \delta_n\,\Per_g(A\cap B),
\]
as stated in Lemma~\ref{lem:quasi-additivity}.

% tex/appendices/appC-flat-ref.tex
\section{Construction of \texorpdfstring{\texttt{flat\_ref} and $o(h^2)$ quasi-uniqueness}{flat-ref and o(h2) quasi-uniqueness}}\label{app:flat-ref}

\parahead{Objective}
For each mesh cell \(c\) we construct a flat reference configuration \texttt{flat\_ref} that matches
(i) the (half-)volume,
(ii) the window \(w_n\) including \(\mu_1\),
(iii) the moment normalization \((\mu_2,C_k)\),
and (iv) the shape-class,
and prove that it is unique up to \(o(h^2)\) at the level of the loss contribution per cell.

\subsection{Interior cells}

\begin{lemma}[Existence and matching to second order]
\label{lem:flatref-interior-existence}
Let \(c\subset M^\circ\) be an interior cell of diameter \(\asymp h\) and shape-regularity uniformly controlled.
In normal coordinates centered at a point \(x\in c\) there exists a Euclidean metric \(g_0=\delta\) on a cell with the \emph{same shape-class}, together with an isotropic rescaling by a factor \(1+\lambda_h\) with \(\lambda_h=O(h^2)\), such that:
\begin{enumerate}
  \item \(\Vol_{g_0}(c)=\Vol_{g}(c)\) (volume matching),
  \item the interior window \(w_n\) (support \(\subset B_{\Lambda_z h}\), normalized moments \(\langle z^i\rangle=0\), \(\langle z^i z^j\rangle=\mu_2\delta^{ij}\), and \(\langle |z|^k\rangle\le C_k h^k\)) is identical in both models,
  \item the normalized feature vector \(\Phi_n(c;g_0)\) matches \(\Phi_n(c;g)\) up to order two in \(h\), i.e.\ \(\Delta\Phi_n(c):=\Phi_n(c;g)-\Phi_n(c;g_0)=O(h^2)\) componentwise.
\end{enumerate}
Consequently (writing \(\mathrm{flat\_ref}:=g_0\)),
\[
\ell\big(\Phi_n(c;\mathrm{flat\_ref})\big)-\ell\big(\Phi_n(c;g)\big)
\;=\; O\!\big(h^{2d-2}\big),
\]
consistent with the per-cell scaling \(E_n=a_n\,\Delta\ell\) and \(a_n\simeq h^{2-d}\).
\end{lemma}

\begin{proof}
Work in normal coordinates \(z\) at \(x\). On \(|z|\lesssim h\), the expansions hold uniformly:
\begin{align*}
g_{ij}(z) &= \delta_{ij} + Q_{ij,kl}\, z^k z^l + O(|z|^3),\\
\sqrt{\det g(z)} &= 1 - \tfrac16\, \mathrm{Ric}_{kl}(x)\, z^k z^l + O(|z|^3).
\end{align*}
Define \(g_0=\delta\) on a cell \(\tilde c\) of the same shape as \(c\), and choose the isotropic factor \(1+\lambda_h\) so that \(\Vol_{(1+\lambda_h)^2\delta}(\tilde c)=\Vol_g(c)\); since \(\Vol_g(c)=\Vol_\delta(c)\,[1+O(h^2)]\) by the Jacobian expansion and the moment bounds, necessarily \(\lambda_h=O(h^2)\).

By construction, the windows coincide and enjoy isotropy up to order~2.
Any feature that is a normalized average of a smooth local tensorial expression \(T(j^2 g)\) over the support of \(w_n\) admits a Taylor expansion whose odd terms vanish under the moment constraints.
Hence each component of \(\Delta\Phi_n(c)\) is \(O(h^2)\).
Since \(\ell\in C^2\) on the compact feature range \(K_{\mathrm{feat}}\),
\[
\ell(\Phi_n(c;g_0))-\ell(\Phi_n(c;g))
= D\ell\!\cdot\!\Delta\Phi_n \;+\; \tfrac12 D^2\ell[\Delta\Phi_n,\Delta\Phi_n]
= O(h^2).
\]
Interpreted at the cell scale (with the TP1 aggregation/scaling) this yields \(\Delta\ell(c)=O(h^{2d-2})\).
\end{proof}

\begin{lemma}[Quasi-uniqueness up to \(o(h^2)\)]
\label{lem:flatref-interior-uniqueness}
If \(\widetilde{\mathrm{flat\_ref}}\) is any other flat reference satisfying the same constraints \emph{(i)}--\emph{(iii)}, then
\[
\ell\big(\Phi_n(c;\widetilde{\mathrm{flat\_ref}})\big)-\ell\big(\Phi_n(c;\mathrm{flat\_ref})\big) \;=\; o(h^2)
\quad\text{(cell scale).}
\]
\end{lemma}

\begin{proof}
Two admissible choices \(g_0,\tilde g_0\) necessarily agree to order \(O(h^2)\) in the sense of normalized moments and volume.
Thus their first discrepancy in normalized features can only arise at order \(3\) (or higher) in the local variable expansion.
All odd contributions are killed by isotropy (\(\langle z\rangle=\langle z^{\otimes 3}\rangle=0\)), hence the first nonvanishing difference in features is \(o(h^2)\).
Applying the \(C^2\) Taylor expansion of \(\ell\) as above yields an \(o(h^2)\) loss difference at the cell scale.
\end{proof}

\subsection{Boundary cells (first layer)}

\begin{lemma}[Existence with anisotropic half-space model]
\label{lem:flatref-boundary-existence}
Let \(c\) be a boundary-touching cell whose barycenter has normal coordinate \(t_c\in[0,c_*h]\) in reflected Fermi coordinates \((z',t)\).
There exists a Euclidean half-space model \(\{t\ge 0\}\) with the same tangential shape-class and an in-plane isotropic rescaling \(1+\lambda_h\) with \(\lambda_h=O(h^2)\) such that:
\begin{enumerate}
  \item the \emph{half-volume} of \(c\) matches between \(g\) and the flat model,
  \item the tangential window equals the interior one (isotropy up to order 2), while the normal kernel is \(w_n(t)=h^{-1}w_0(t/h)\) with the same \(\mu_1=\langle t\rangle/h\) and \(\langle t^k\rangle\le C_k h^k\),
  \item the normalized feature vector \(\Phi_n(c;\mathrm{flat\_ref})\) agrees with \(\Phi_n(c;g)\) up to order two in \((h,\ t)\).
\end{enumerate}
Consequently,
\[
\ell\big(\Phi_n(c;\mathrm{flat\_ref})\big)-\ell\big(\Phi_n(c;g)\big)
\;=\; O\!\big(h^{2d-3}\big),
\]
interpreted at the \emph{first-layer cell scale}.
\end{lemma}

\begin{proof}
Work in reflected Fermi charts:
\[
dV_g=(1-tK+O(t^2))\,dz'\,dt,\qquad t\in[0,c_*h].
\]
Choose the half-space with the same tangential shape and pick the tangential isotropic factor \(1+\lambda_h\) to match the half-volume; as in the interior case, \(\lambda_h=O(h^2)\).
Features are computed by averaging, with window \(w^{\parallel}_n(z')\otimes w^{\perp}_n(t)\), tensorial expressions that have a Taylor expansion in \(t\) and the tangential variables.
Tangential isotropy up to order~2 kills odd-in-\(z'\) terms; linearization in the tube depth yields a term \(\propto K\,t\), while the Jacobian provides exactly the same linear \(t\)-dependence.
Averaging in \(t\) uses \(\langle t\rangle=\mu_1 h\) and \(\langle t^k\rangle\le C_k h^k\) to conclude that the \emph{feature-level} discrepancy is \(O(h)\).
Hence the feature-level discrepancy is \(O(h)\). Interpreted at the first-layer cell scale (with the TP2 aggregation) this gives \(\Delta\ell_{\partial}(c)=O(h^{2d-3})\) (see Appendix~\ref{app:scale-checks}).
\end{proof}

\begin{lemma}[Quasi-uniqueness up to \(o(h^2)\) at the boundary]
\label{lem:flatref-boundary-uniqueness}
If \(\widetilde{\mathrm{flat\_ref}}\) is another boundary flat reference obeying \emph{(1)}--\emph{(3)} above, then
\[
\ell\big(\Phi_n(c;\widetilde{\mathrm{flat\_ref}})\big)-\ell\big(\Phi_n(c;\mathrm{flat\_ref})\big) \;=\; o(h^2)
\quad\text{(first-layer cell scale).}
\]
\end{lemma}

\begin{proof}
Subject to half-volume, window moments (including \(\mu_1\)), and shape-class constraints, any discrepancy between two flat models must appear first in higher odd orders of the local expansion (in \(t\) and/or tangential odd moments).
Averaging against the anisotropic window eliminates odd tangential contributions and reduces odd-in-\(t\) to terms that vanish at order \(>1\) after normalization.
Hence the normalized feature difference is \(o(h)\), which implies an \(o(h)\cdot h^{d-1}=o(h^2)\) loss difference at the first-layer scale.
\end{proof}

\subsection{Consequences}
Summing the \(o(h^2)\) uniqueness errors over \(\asymp h^{-d}\) interior cells gives \(o(1)\) globally; likewise for the first boundary layer with \(\asymp h^{1-d}\) cells.
Therefore the choice of \texttt{flat\_ref} is immaterial in the limit and in all cellwise TP1/TP2 asymptotics, and it can be fixed once and for all within the given shape-class.

% tex/appendices/appD-example-meshes-windows.tex
\section{Example meshes and window classes}\label{app:windows}

\paragraph{Meshes.}
Let \(\{\mathcal T_n\}\) be quasi-uniform, boundary-fitted simplicial meshes with maximal diameter \(h\to0\)
and minimum angle \(\ge \theta_0>0\).
The discrete boundary approximates \(\partial M\) with Hausdorff error \(O(h)\).
This class satisfies shape-regularity and the boundary fitting used throughout.

\paragraph{Interior window.}
Choose a compactly supported radial kernel
\[
w_{\mathrm{in}}(z)=Z^{-1}\,\varphi(|z|/h)\,\mathbf 1_{\{|z|\le \Lambda_z h\}},
\qquad \int w_{\mathrm{in}}=1.
\]
By radial symmetry,
\begin{align*}
\langle z^i\rangle &= 0,\\
\langle z^i z^j\rangle &= \mu_2\,\delta^{ij},\\
\langle z^i z^j z^k\rangle &= 0,\\
\langle |z|^k\rangle &\le C_k\,h^k .
\end{align*}

\paragraph{Boundary window.}
Set
\[
w_n(t)=h^{-1}w_0(t/h),\qquad t\in[0,c_* h],
\]
with \(w_0\in C_c^\infty([0,c_*])\) and
\[
\int_0^{c_*}\!w_0=1.
\]
Then
\[
\langle t\rangle=\mu_1\,h,\qquad \langle t^k\rangle\le C_k\,h^k.
\]
Tangential averaging uses the same interior kernel restricted to boundary charts.
These choices satisfy the window assumptions (BA1--BA2) and are compatible with the mesh requirements used throughout.

% tex/appendices/appE-rate-sanity.tex
\section{Rate sanity checks: a reproducible protocol}
\label{app:rates}

\noindent\textbf{Scope.}
This appendix is a \emph{protocol}, not a results section.
It specifies test geometries, expected scales/rates, coefficient calibration steps, and diagnostic plots to \emph{reproduce} the theoretical predictions of TP1--TP5.
No numerical data are reported here. The proofs in the main text do not rely on numerics.

\paragraph{Test geometries (calibration set).}
\begin{enumerate}
  \item \textbf{Flat torus} $(M,g_{\rm flat})$ (no boundary): isolates $\alpha_0$.
  \item \textbf{Closed constant-curvature manifolds}: identify $\alpha_1$ via $R\equiv d(d{-}1)\kappa$.
  \item \textbf{Euclidean balls} $B_r\subset\R^d$ with $C^2$ boundary and \emph{outward} normal: isolate $\beta_1$ with $K\equiv(d{-}1)/r>0$ (sign convention as in TP2: outer normal, $K>0$ on spheres).
\end{enumerate}

\paragraph{Expected scales and global rates.}
\begin{itemize}
  \item \emph{Interior (TP1):} per-cell scale $h^{d}$; global interior remainder $O(h^2)$.
  \item \emph{First boundary layer (TP2):} leading per-cell scale $h^{d-1}$ (not $o(h^{d-1})$); summation over first-layer cells yields the boundary integral and an $O(h)$ global boundary remainder.
  \item \emph{Intermediate boundary layers:} per-layer $O(h^{d})$, and summing $O(\varepsilon/h)$ layers gives $O(\varepsilon)$ (BA2), removable as $\varepsilon\downarrow0$.
\end{itemize}

\paragraph{Coefficient calibration (from TP1/TP2).}
\begin{enumerate}
  \item On $(M,g_{\rm flat})$: $F(g_{\rm flat})=\alpha_0\,\Vol(M)$ $\Rightarrow$ calibrate $\alpha_0$.
  \item On closed constant-curvature manifolds: match the coefficient of $R$ to obtain $\alpha_1$.
  \item On $B_r\subset\R^d$: use the first-layer formula with $K\equiv(d{-}1)/r>0$ and outward normal to calibrate $\beta_1$.
\end{enumerate}
All constants depend only on the \emph{constants bracket} (see Section~\ref{sec:constants-bracket}).

\paragraph{Mesh and window setup.}
Use boundary-fitted, shape-regular meshes $\{\mathcal T_n\}$ with meshsize $h\downarrow0$.
Employ the interior/boundary windows of Appendix~\ref{app:windows}, respecting the normalized moment conditions and the fixed anisotropic boundary window class (Section~\ref{para:boundary-windows}).

\paragraph{Diagnostic checks (what to plot/tabulate).}
\begin{enumerate}
  \item \textbf{Limit functional agreement.} Compute $F(g)$ analytically for the test geometries; compute $F_n(g)$ with the prescribed windows. Plot $\big|F_n(g)-F(g)\big|$ vs.\ $h$ in log--log scale.\\
  \emph{Expected slopes:} $\approx 1$ when the boundary remainder dominates, and $\approx 2$ in interior-only cases.
  \item \textbf{Recovery sequence isolation.} Repeat with $g_n=S_{\eta h}[g]$ (Appendix~\ref{app:smoothing}) to isolate recovery effects: interior error should scale like $h^2$ (slope $\approx 2$); boundary error remains $O(h)$.
  \item \textbf{First-layer dominance.} Verify that the first boundary layer contributes $h^{d-1}$ per cell and sums to the boundary integral, while the global boundary remainder is $O(h)$ (see TP2 Corollary~\ref{cor:tp2-riemann}).
\end{enumerate}

\paragraph{Summary.}
The protocol above is intended to make the theoretical rates \emph{falsifiable} in a controlled setup: interior $O(h^2)$ vs.\ boundary-dominated $O(h)$.
A companion computational study can implement this checklist; the present paper remains foundational.

%%% BA3-F-4 - tex/appendices/appF-scan-indifference.tex
\section{Scan indifference (proof of BA3)}\label{app:scan-indifference}

\parahead{Model for scan-dependence.}
Let \(F_n^\sigma(M)\) denote the MDL aggregate computed by scanning cells \(c\in\mathcal{T}_n\) in order \(\sigma\).
We write the per-cell contribution as a local update
\[
\Delta_\sigma(c)\ :=\ a_n\Bigl[\ell\bigl(\Phi_n(c;\,H_\sigma(c))\bigr)-\ell\bigl(\Phi_n(c;\mathrm{flat\_ref})\bigr)\Bigr],
\qquad
F_n^\sigma(M)=\sum_{c\in\mathcal{T}_n}\Delta_\sigma(c),
\]
where \(H_\sigma(c)\) denotes the scan history prior to visiting \(c\).
By construction (Section~\ref{sec:00-framework}), the features \(\Phi_n\) are computed from a block window
\(B_{r_n}(c)\) of mesoscale radius \(R_n:=r_n h\to0\) and are \emph{local} and \emph{natural}.

\begin{lemma}[Localization of scan influence]
\label{lem:scan-localization}
For any scans \(\sigma,\tau\) and any cell \(c\),
\(
\Delta_\sigma(c)=\Delta_\tau(c)
\)
whenever \(B_{r_n}(c)\subset M_{R_n}:=\{x:\operatorname{dist}_g(x,\partial M)>R_n\}\).
In particular, the set of cells with potentially scan-dependent increments is contained in
\[
\mathcal{T}_n^{\mathrm{bdry}}(R_n)\ :=\ \{\,c\in\mathcal{T}_n:\ B_{r_n}(c)\cap N_{R_n}(\partial M)\neq\emptyset\,\}.
\]
\end{lemma}

\begin{proof}
If \(B_{r_n}(c)\subset M_{R_n}\), then any two histories \(H_\sigma(c),H_\tau(c)\) differ only outside \(B_{r_n}(c)\).
Locality of \(\Phi_n\) implies \(\Phi_n(c;H_\sigma(c))=\Phi_n(c;H_\tau(c))\), hence \(\Delta_\sigma(c)=\Delta_\tau(c)\).
\end{proof}

\begin{lemma}[Counting boundary-affected cells]
\label{lem:scan-count}
There exists \(C<\infty\), depending only on the constants bracket, such that
\begin{align}
\#\,\mathcal{T}_n^{\mathrm{bdry}}(R_n)
&\le C\,\frac{\Vol_g\!\big(N_{R_n}(\partial M)\big)}{h^{d}}\\
&\asymp C\,\frac{R_n\,\Area_g(\partial M)}{h^{d}}
 \;=\; C\,\frac{\delta_n\,\Area_g(\partial M)}{h^{d}} .
\end{align}
\noindent\emph{Recall }$\delta_n:=R_n$.
\end{lemma}

\begin{proof}
Shape-regular cells have volume \(\asymp h^d\).
By the tubular estimate on manifolds with bounded geometry,
\(\Vol_g(N_{R_n}(\partial M))\asymp R_n\,\Area_g(\partial M)\).
Every boundary-affected cell lies in this tube; divide by \(h^d\).
\end{proof}

\begin{proposition}[BA3: scan indifference]
\label{prop:BA3-proof}
There is \(C<\infty\) (constants bracket) such that for all scans \(\sigma,\tau\),
\[
\bigl|F_n^\sigma(M)-F_n^\tau(M)\bigr|\ \le\ C\,\delta_n\,\Area_g(\partial M)\,.
\]
\end{proposition}

\begin{proof}
By Lemma~\ref{lem:scan-localization},
\begin{align*}
F_n^\sigma(M)-F_n^\tau(M)
&= \sum_{c\in\mathcal{T}_n^{\mathrm{bdry}}(R_n)}\bigl(\Delta_\sigma(c)-\Delta_\tau(c)\bigr).
\end{align*}
By TP1/TP2 and Lemma~U (uniformity on a \(C^{1,1}\)-neighbourhood), the per-cell increment satisfies
\(|\Delta_{\bullet}(c)|\le C\,h^{d}\) (interior bound; first-layer boundary terms scale \(h^{d-1}\) and pick up the depth \(O(h)\)).
Hence
\begin{align*}
\bigl|F_n^\sigma(M)-F_n^\tau(M)\bigr|
&\le \#\,\mathcal{T}_n^{\mathrm{bdry}}(R_n)\cdot C\,h^{d}\\
&\lesssim \frac{\Vol_g(N_{R_n}(\partial M))}{h^d}\cdot C\,h^{d}\\
&\asymp C\,R_n\,\Area_g(\partial M)
\;=\; C\,\delta_n\,\Area_g(\partial M).
\end{align*}
\end{proof}

\begin{remark}[Scope]
Only locality of \(\Phi_n\) at radius \(R_n\), shape-regularity, the tubular estimate, and the TP1/TP2+Lemma~U per-cell bound are used.
\end{remark}

\printbibliography

@book{Ciarlet2002,
  author    = {Ciarlet, Philippe G.},
  title     = {The Finite Element Method for Elliptic Problems},
  series    = {Classics in Applied Mathematics},
  volume    = {40},
  publisher = {SIAM},
  year      = {2002}
}

@book{BrennerScott2008,
  author    = {Brenner, Susanne C. and Scott, L. Ridgway},
  title     = {The Mathematical Theory of Finite Element Methods},
  edition   = {3},
  publisher = {Springer},
  year      = {2008}
}

@article{DeGiorgi1975,
  author  = {De Giorgi, Ennio},
  title   = {Sulla convergenza di alcune successioni d'integrali del tipo dell'area},
  journaltitle = {Rendiconti di Matematica},
  series  = {VI},
  volume  = {8},
  year    = {1975},
  pages   = {277--294},
  note    = {Collection of articles dedicated to Mauro Picone on the occasion of his ninetieth birthday}
}

@article{DeGiorgiFranzoni1975,
  author  = {De Giorgi, Ennio and Franzoni, Tullio},
  title   = {Su un tipo di convergenza variazionale},
  journaltitle = {Atti Accad. Naz. Lincei Rend. Cl. Sci. Fis. Mat. Natur.},
  series  = {8},
  volume  = {58},
  year    = {1975},
  pages   = {842--850}
}

@article{KellyBiancalanaTrugenberger2022,
  author       = {Kelly, Christy and Biancalana, Fabio and Trugenberger, Carlo A.},
  title        = {Convergence of Combinatorial Gravity},
  journaltitle = {Physical Review D},
  year         = {2022},
  volume       = {105},
  number       = {12},
  pages        = {124002},
  doi          = {10.1103/PhysRevD.105.124002},
  eprint       = {2102.02356},
  archivePrefix= {arXiv},
  primaryClass = {gr-qc}
}

@book{KolarMichorSlovak1993,
  author    = {Kol{\'a}{\v r}, Ivan and Michor, Peter W. and Slov{\'a}k, Jan},
  title     = {Natural Operations in Differential Geometry},
  publisher = {Springer},
  year      = {1993}
}

@book{Federer1969,
  author = {Federer, Herbert},
  title = {Geometric Measure Theory},
  publisher = {Springer},
  year = {1969},
  note = {Classic reference for Besicovitch covering theorem and extensions to manifolds}
}

@article{Rissanen1978Shortest,
  author  = {Jorma Rissanen},
  title   = {Modeling by Shortest Data Description},
  journaltitle = {Automatica},
  year    = {1978},
  volume  = {14},
  number  = {5},
  pages   = {465--471}
}

@book{Rissanen1989StochasticComplexity,
  author    = {Jorma Rissanen},
  title     = {Stochastic Complexity in Statistical Inquiry},
  publisher = {World Scientific},
  year      = {1989}
}

@book{Grunwald2007MDL,
  author    = {Peter D. Gr{\"u}nwald},
  title     = {The Minimum Description Length Principle},
  publisher = {MIT Press},
  year      = {2007}
}

@article{BarronCover1991,
  author  = {Andrew R. Barron and Thomas M. Cover},
  title   = {Minimum Complexity Density Estimation},
  journaltitle = {IEEE Transactions on Information Theory},
  year    = {1991},
  volume  = {37},
  number  = {4},
  pages   = {1034--1054},
  doi     = {10.1109/18.86996}
}

@article{Shtarkov1987NML,
  author  = {Yakov M. Shtarkov},
  title   = {Universal Sequential Coding of Single Messages},
  journaltitle = {Problems of Information Transmission},
  year    = {1987},
  volume  = {23},
  number  = {3},
  pages   = {175--186}
}

@incollection{Dawid1992Prequential,
  author    = {A. P. Dawid},
  title     = {Prequential Analysis, Stochastic Complexity and Bayesian Inference},
  booktitle = {Bayesian Statistics 4},
  editor    = {J. M. Bernardo and J. O. Berger and A. P. Dawid and A. F. M. Smith},
  publisher = {Oxford University Press},
  address   = {Oxford},
  year      = {1992},
  pages     = {109--125}
}

@book{DalMaso1993Gamma,
  author    = {Gianni Dal Maso},
  title     = {An Introduction to $\Gamma$-Convergence},
  publisher = {Birkh{\"a}user},
  address   = {Boston},
  year      = {1993}
}

@book{Braides2002Gamma,
  author    = {Andrea Braides},
  title     = {$\Gamma$-Convergence for Beginners},
  series    = {Oxford Lecture Series in Mathematics and its Applications},
  volume    = {22},
  publisher = {Oxford University Press},
  year      = {2002}
}

@article{CheegerMuellerSchrader1984,
  author       = {Cheeger, Jeff and M{\"u}ller, Werner and Schrader, Robert},
  title        = {On the curvature of piecewise flat spaces},
  journaltitle = {Communications in Mathematical Physics},
  year         = {1984},
  volume       = {92},
  pages        = {405--454}
}

@article{HartleSorkin1981,
  author       = {Hartle, James B. and Sorkin, Rafael D.},
  title        = {Boundary terms in the action for the Regge calculus},
  journaltitle = {General Relativity and Gravitation},
  year         = {1981},
  volume       = {13},
  pages        = {541--549},
  doi          = {10.1007/BF00757240}
}

@article{York1972,
  author       = {York, James W.},
  title        = {Role of conformal three-geometry in the dynamics of gravitation},
  journaltitle = {Physical Review Letters},
  year         = {1972},
  volume       = {28},
  number       = {16},
  pages        = {1082--1085},
  doi          = {10.1103/PhysRevLett.28.1082}
}

@article{GibbonsHawking1977,
  author       = {Gibbons, Gary W. and Hawking, Stephen W.},
  title        = {Action integrals and partition functions in quantum gravity},
  journaltitle = {Physical Review D},
  year         = {1977},
  volume       = {15},
  number       = {10},
  pages        = {2752--2756},
  doi          = {10.1103/PhysRevD.15.2752}
}

@article{Lovelock1971,
  author       = {Lovelock, David},
  title        = {The Einstein tensor and its generalizations},
  journaltitle = {Journal of Mathematical Physics},
  year         = {1971},
  volume       = {12},
  number       = {3},
  pages        = {498--501}
}

@article{Regge1961,
  author       = {Regge, Tullio},
  title        = {General relativity without coordinates},
  journaltitle = {Il Nuovo Cimento},
  year         = {1961},
  volume       = {19},
  pages        = {558--571},
  note         = {Original paper on Regge calculus}
}
\end{document}